\documentclass{amsart}

\newcommand{\NL}{}
\usepackage[top=.75 in,bottom=.75 in, left=.75 in, right= .75 in]{geometry}  \renewcommand{\NL}{ }
\newcommand{\lbl}[1]{\label{#1}}

\usepackage{color,soul}

 \usepackage{tikz}
 \usetikzlibrary{patterns}

 \usepackage[numbers]{natbib}
\usepackage{graphicx,fancybox,pgf,amsmath,float,amssymb}
 \usepackage{framed}
 \usepackage{color,soul}

\newcommand{\eps}{\varepsilon}

\newcommand{\Binq}[2]{\begin{bmatrix}#1\\#2
\end{bmatrix}_q
}

\newcommand{\calA}{{\mathcal M}}

\numberwithin{equation}{section}

\usepackage{dsfont}

\newcommand{\CC}{\mathds{C}}
\newcommand{\NN}{\mathds{N}}

\newcommand{\mA}{\mathbf{A}}

\newcommand{\mE}{\mathbf{E}}
\newcommand{\mD}{\mathbf{D}}
\newcommand{\mI}{\mathbf{I}}

\newcommand{\mX}{\mathbf{X}}
\newcommand{\mY}{\mathbf{Y}}
\newcommand{\mZ}{\mathbf{Z}}
\newcommand{\mS}{\mathbf{S}}

\newcommand{\xx}{\mbox{\fontfamily{phv}\selectfont x}}

\newcommand{\dd}{\mbox{\fontfamily{phv}\selectfont d}}
\newcommand{\ee}{\mbox{\fontfamily{phv}\selectfont e}}
\newcommand{\Dq}{\mbox{\fontfamily{phv}\selectfont D}_q}
\newcommand{\Z}{\mbox{\fontfamily{phv}\selectfont Z}}

\newcommand{\SSS}{\mbox{\fontfamily{phv}\selectfont S}}

\newtheorem{theorem}{Theorem}
 \newtheorem{lemma}{Lemma}
\newtheorem{proposition}{Proposition}
\newtheorem{corollary}{Corollary}
\theoremstyle{definition}

\theoremstyle{remark}
\newtheorem{remark}{Remark}
\theoremstyle{remark}

 \usepackage{framed}
\colorlet{shadecolor}{gray!30}
\newenvironment{ana}
  {\begin{leftbar}
  \begin{shaded} }
{  \end{shaded}\end{leftbar}}
\newcounter{oldeq}
\newcounter{usesofarxiv}

 \newcommand{\arxiv}[1]{
 \setcounter{oldeq}{\value{equation}}
  \addtocounter{usesofarxiv}{1}
   \setcounter{equation}{0}
\def\theoldeq{\theequation}
\def\theequation{\arabic{section}.\arabic{oldeq}.\arabic{equation}}
 \begin{ana}
  \footnotesize #1\end{ana}
   \setcounter{equation}{\value{oldeq}}
\numberwithin{equation}{section}
     }

\subjclass[2010]{82C22;60K35;33D45;33D15}
\title[ASEP in the singular case]{On matrix product ansatz for  Asymmetric Simple Exclusion Process with open boundary in the singular case}

\author{W\l odzimierz Bryc}
\address
{
W\l odzimierz Bryc\\
Department of Mathematical Sciences\\
University of Cincinnati\\
2815 Commons Way\\
Cincinnati, OH, 45221-0025, USA.
}
\email{wlodek.bryc@gmail.com}
\author{Marcin \'Swieca}
\address{Marcin \'Swieca\\
Department of Mathematical Sciences\\
University of Cincinnati\\
2815 Commons Way\\
Cincinnati, OH, 45221-0025, USA.
and  Faculty of Mathematics and Information Science\\
Warsaw University of Technology\\ pl. Politechniki 1 00-661\\
Warszawa, Poland
}
\email{marcin.swieca@mini.pw.edu.pl}

\begin{document}

\maketitle
\keywords{Asymmetric simple exclusion process with open boundary;  Askey-Wilson polynomials; matrix product ansatz}

\begin{abstract}
We study a substitute for the matrix product ansatz for Asymmetric Simple Exclusion Process with open boundary in the ``singular case'' $\alpha\beta=q^N\gamma\delta$,  when the standard form of the matrix  product ansatz of
Derrida, Evans, Hakim and Pasquier [J. Phys. A 26(1993)] does not apply.
 In our approach, the matrix product ansatz is replaced with a pair of  linear functionals  on an abstract algebra. One of
      the functionals,   $\varphi_1$, is defined on the entire algebra, and  determines stationary probabilities for large systems on $L\geq N+1$ sites. The other functional,   $\varphi_0$, is defined only on  a
      finite-dimensional linear subspace  of the algebra, and  determines stationary probabilities for small  systems on $L< N+1$ sites.
Functional $\varphi_0$  vanishes on non-constant Askey-Wilson polynomials and in non-singular case  becomes an orthogonality functional for the Askey-Wilson polynomials.
\end{abstract}
\arxiv{This is an expanded version of the paper. It includes additional material that is typeset differently from the main body
of the paper.}

\section{Introduction and main results}
The  Asymmetric Simple Exclusion Process  (ASEP) with open boundary on sites $\{1,\dots,L\}$  is a  continuous time Markov chain with state space $\{0,1\}^L$.
Informally, see
  Fig. \ref{Fig1},   particles may arrive at the left boundary  at rate $\alpha>0$ and leave at rate $\gamma\geq 0$.
A particle may move to the right at rate $1$ or to the left at rate $q<1$.  It may leave at the right boundary at rate $\beta>0$  or a new particle may arrive there at rate $\delta\geq0$. At most one particle is allowed at each site.
More formal description of the evolution is given as Kolmogorov's equations \eqref{DiffEq} below.
\begin{figure}[H]

  \begin{tikzpicture}[scale=.8]
\draw [fill=black, ultra thick] (.5,1) circle [radius=0.2];
  \draw [ultra thick] (1.5,1) circle [radius=0.2];
\draw [fill=black, ultra thick] (2.5,1) circle [radius=0.2];
  \draw [ultra thick] (5,1) circle [radius=0.2];
   \draw [fill=black, ultra thick] (6,1) circle [radius=0.2];

    \draw [ultra thick] (7,1) circle [radius=0.2];
      \draw [fill=black, ultra thick] (9.5,1) circle [radius=0.2];
   \draw [ultra thick] (10.5,1) circle [radius=0.2];
     \draw[->] (-1,2.3) to [out=-20,in=135] (.5,1.5);
   \node [above right] at (-.2,2) {$\alpha$};
     \draw[->] (10.5,1.5) to [out=45,in=200] (12,2.3);
     \node [above left] at (11.2,2) {$\beta$};
            \node  at (8.25,1) {$\cdots$};  \node  at (3.75,1) {$\cdots$};
      \node [above] at (6.5,1.8) {$1$};
      \draw[->,thick] (6.1,1.5) to [out=45,in=135] (7,1.5);
        \node [above] at (5.5,1.8) {$q$};
            \draw[<-] (5,1.5) to [out=45,in=135] (5.9,1.5);
                 \node [above] at (10,1.8) {$1$};
                \draw[->,thick] (9.6,1.5) to [out=45,in=135] (10.4,1.5);
               \node [above] at (9,1.8) {$q$};
                 \draw[<-] (8.4,1.5) to [out=45,in=135] (9.4,1.5);
       \draw[<-] (-1,-.3) to [out=0,in=-135] (.5,0.6);
   \node [below right] at (-.2,0) {$\gamma$};
    \node [above] at (0.5,0) {$1$};
    \node [above] at (1.5,0) {$2$};
   \node [above] at (2.5,0) {$3$};
     \node [above] at (3.75,0) {$\cdots$};
     \node [above] at (6,0) {\textcolor{white}{$+1$} $k$ \textcolor{white}{$+1$}};
        \node [above] at (8.25,0) {$\cdots$};
         \node [above] at (10.5,0) {\textcolor{white}{$+$}$L$\textcolor{white}{$1$}};
          \node [above] at (9.5,0) {$L-1$};
        \draw[<-] (10.6,.7) to [out=-45,in=180] (12,-.3);
   \node [below left] at (11.2,0) {$\delta$};
\end{tikzpicture}
 \caption{Asymmetric simple exclusion process (ASEP) on $\{1,\dots,L\}$  with open boundaries,  with parameters $\alpha,\beta>0$,
$\gamma,\delta\geq 0$, and  $0\leq q<1$. %
Filled in disks represent occupied sites.
\lbl{Fig1}
}
\end{figure}
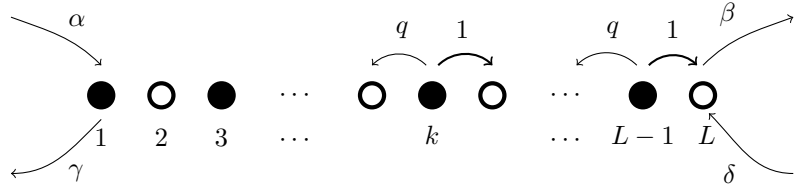

We are interested in the steady state of the ASEP, so we focus on the stationary distribution %
of the Markov chain. The standard method relies on Kolmogorov's prospective equations. Denoting by $P_t(\tau_1,\dots,\tau_L)$  the probability that Markov chain is in configuration $(\tau_1,\dots,\tau_L)\in
\{0,1\}^L$ at time t, we have
\begin{multline}\label{DiffEq}
\frac{d}{dt}P_t(\tau_1,\dots,\tau_L)= \delta_{\tau_1=1}\left[\alpha P_t(0,\tau_2,\dots,\tau_L)-\gamma  P_t(1,\tau_2,\dots,\tau_L)\right]\NL
+\delta_{\tau_1=0}\left[\gamma P_t(1,\tau_2,\dots,\tau_L)-\alpha  P_t(0,\tau_2,\dots,\tau_L)\right]\\
+\sum_{k=1}^{L-1} \delta_{\tau_k=1,\tau_{k+1}=0}\Big[qP_t(\tau_1,\dots,\tau_{k-1},0,1,\tau_{k+2},\dots,\tau_L)
\NL- P_t(\tau_1,\dots,\tau_{k-1},1,0,\tau_{k+2},\dots,\tau_L)
\Big] \\
+\sum_{k=1}^{L-1} \delta_{\tau_k=0,\tau_{k+1}=1}\Big[P_t(\tau_1,\dots,\tau_{k-1},1,0,\tau_{k+2},\dots,\tau_L)
\NL- q P_t(\tau_1,\dots,\tau_{k-1},0,1,\tau_{k+2},\dots,\tau_L)
\Big]\\
+\delta_{\tau_L=0}\left[\beta P_t(\tau_1,\dots,\tau_{L-1},1)-\delta P_t(\tau_1,\dots,\tau_{L-1},0)\right] \NL
+\delta_{\tau_L=1}\left[\delta P_t(\tau_1,\dots,\tau_{L-1},0)-\beta P_t(\tau_1,\dots,\tau_{L-1},1)\right].
\end{multline}
The stationary distribution $P(\tau_1,\dots,\tau_L)$ of this Markov chain  satisfies
 $$\frac{d}{dt}P_t(\tau_1,\dots,\tau_L)=0$$ so it solves the system of linear equations on the right hand side of \eqref{DiffEq}.
 An ingenious method of determining the stationary probabilities for all $L$  was introduced by Derrida, Evans, Hakim and Pasquier in
\cite{derrida1993exact}, who consider
  infinite  matrices and  vectors that satisfy relations
 \begin{eqnarray}
  \mD\mE-q\mE\mD&=&\mD+\mE ,\label{q-comm-Derrida}\\
\langle W|(\alpha \mE-\gamma \mD)&=&\langle W| ,\label{W}\\
(\beta \mD-\delta \mE)|V\rangle&=&|V\rangle. \label{V}
\end{eqnarray}
The   stationary  probabilities are then computed  as
\begin{equation}
  \label{MatrixSolution}
 P(\tau_1,\dots,\tau_L)= \frac{\langle W|\prod_{j=1}^L (\tau_j\mD+(1-\tau_j)\mE)|V\rangle }{\langle W|(\mD+ \mE)^L|V\rangle}.
\end{equation}

It has been noted in the literature that the above approach may fail: \citet[page 3384]{essler1996representations} point out   that  matrix representation \eqref{MatrixSolution}     runs into problems
when
$\alpha\beta=\gamma\delta$,
and they point out the importance of a more general condition that $\alpha\beta-q^n\gamma\delta\ne0$ for $n=0,1,\dots$. We will call this a non-singular case.

The singular case when   $\alpha\beta=q^N\gamma\delta$, is
  discussed  by
 \citet[Appendix A]{Mallick-Sandow-1997} in the context of finite matrix representations.
 Of course, this   is a singular case for the matrix product ansatz,   not for the actual Markov chain.
 To avoid singularity,  \citet{lazarescu2013matrix} presents a perturbative generalization of the matrix product ansatz,
which was used in \citep{gorissen2012exact} to derive exact current statistics for all values of parameters.  Continuity  of the ASEP with respect to its parameters is also used to
derive recursion  for stationary probabilities in \cite[proof of Theorem 2.3]{Liggett-1975}.

\subsection{Solution for the singular case}
 Our goal is to
 analyze  the singular case $\alpha\beta=q^N\gamma\delta$ directly.
  We consider an abstract noncommutative algebra $\calA$ with identity
$\mI$ and two generators $\mD,\mE$ that satisfy relation \eqref{q-comm-Derrida}.
 The algebra consists of linear combinations of monomials
 $\mX=\mD^{n_1}\mE^{m_1}\dots \mD^{n_k}\mE^{m_k}$.   It turns out that monomials in normal order, $\mE^m\mD^n$,  form a basis for $\calA$ as a vector
space.
We   introduce increasing subspaces $\calA_k$ of $\calA$ that are spanned by the monomials  in normal order of degree at most $k$,  i.e., $\calA_k$ is the span of $\{\mE^m\mD^n: m+n\leq k\}$.
The abstract version of the matrix product ansatz  for the singular case    uses a pair of
linear functionals $\varphi_0:\calA_N\to\CC$ and $\varphi_1:\calA\to\CC$. %

\begin{theorem} \label{T-0}
Suppose $\alpha,\beta, \gamma,\delta>0$ satisfy  $\alpha \beta=q^N \gamma\delta$ for some  $N=0,1,\dots$. Then there exists a  pair of linear functionals $\varphi_0:\calA_N\to\CC$
and $\varphi_1:\calA\to\CC$ such that stationary probabilities for the ASEP are
 \begin{equation}\label{phi2P}
 P(\tau_1,\dots,\tau_L)= \frac{\varphi\left[\prod_{j=1}^L (\tau_j\mD+(1-\tau_j)\mE)\right]}{\varphi\left[(\mD+ \mE)^L\right]},
 \end{equation}
 where $\varphi=\varphi_0$ if $1\leq L< N+1$ and $\varphi=\varphi_1$ if $L\geq N+1$.
 Furthermore,
  if $L=N+1$   then the stationary  distribution is the product of Bernoulli measures
    $$P(\tau_1,\dots,\tau_{N+1}) = \prod_{j=1}^{N+1} p_j^{\tau_j} q_j^{1-\tau_j}$$
with $p_j=\frac{\alpha}{\alpha+\gamma q^{j-1}}$
and $q_j=1-p_j$.

If  $\alpha,\beta>0$,  $\gamma,\delta\geq 0$ are such that  $\alpha \beta\ne q^n \gamma\delta$ for all  $n=0,1,\dots$, then $\varphi_0$ is defined on $\calA_\infty=\calA$, and \eqref{phi2P} holds with $\varphi=\varphi_0$ for all $L$.
\end{theorem}
We remark that part of the conclusion of the theorem is the assertion that the denominators in \eqref{phi2P} are  non-zero for all $L$. Proposition \ref{P-positivity} below   determines their signs, which according to Remark \ref{Rem-signs} may vary also in the non-singular case.
The signs determine the direction of the {\em current} $J$ through the bond between adjacent sites, which is defined as $J=\Pr(\tau_k=1,\tau_{k+1}=0)-q\Pr(\tau_k=0,\tau_{k+1}=1)$. When $L\ne N+1$, we have $J=\varphi[(\mE+\mD)^{L-1}]/\varphi[(\mE+\mD)^{L}]$, so the current is negative  for $2\leq L\leq N$, and  positive for $L> N+1$.
As noted in \citep[Section 3]{Aneva-2009-Integrability}, the current    vanishes for $L=N+1$ due to the detailed balance condition satisfied by the product measure.

The proof of  Theorem \ref{T-0} is given in Section \ref{sect:ProofT1} and consist  of recursive construction of the pair of functionals.
In the construction, the left and right eigenvectors in \eqref{W} and \eqref{V} are replaced by the left and right invariance
requirements:
\begin{equation}
  \label{W+} \varphi\left[(\alpha \mE-\gamma\mD)\mA\right]=\varphi[\mA],
\end{equation}
\begin{equation}
  \label{V+} \varphi\left[\mA(\beta \mD-\delta\mE)\right]=\varphi[\mA],
\end{equation}
for all $\mA\in\calA$ when $\varphi=\varphi_1$ and for all $\mA\in\calA_{N-1}$ if $\varphi=\varphi_0$.
By an adaptation of the argument from \cite{derrida1993exact}, functionals that satisfy \eqref{W+} and \eqref{V+} give stationary probabilities, see Theorem \ref{T2} for precise statement. Similar  modification of \eqref{W} and \eqref{V} in the  matrix formulation appears in  \cite[Theorem 5.2]{corteel2011tableaux}.
After the paper was submitted, we learned that the idea of working with an abstract algebra and defining a linear functional by
using normal order can be traced back to    \cite[Section 3]{derrida-mallick-97} who consider periodic ASEPs, so constraints
\eqref{W+} and \eqref{V+} do not appear. %

In the singular case functional $\varphi_0$ is defined on $N(N+1)/2$-dimensional space $\calA_N$.   However, $\calA_N$ is not
an algebra,
 so this is different from
  the finite dimensional
 representations of the matrix algebra which  were studied by
   \citet{essler1996representations} and \citet{Mallick-Sandow-1997}.
   In Appendix \ref{Sec:MatrixModel} we present a ``matrix model" for all $\alpha,\beta,\gamma,\delta$ with $0<q<1$
   that was  inspired by \citet{Mallick-Sandow-1997}. The model  reproduces their
    finite matrix model   when the parameters are chosen like in their paper, but cannot be used for general parameters  due to lack of associativity.

\subsection{Relation to Askey-Wilson polynomials}
 Ref. \cite{uchiyama2004asymmetric} shows that the stationary distribution of the open ASEP is intimately related
 to the %
  Askey-Wilson polynomials. Here we extend this relation to cover also the singular case,  when the Askey-Wilson polynomials   do not have the Jacobi matrix, see discussion below.

In the context of   ASEP, the Askey-Wilson  polynomials depend on parameter $q$, %
and on four real parameters $a,b,c,d$  which
are related to parameters of ASEP by the  equations
\begin{multline}
  \label{AW-parameters}
  \alpha=\frac{1-  q}{(1+c)(1+d)}, \beta=\frac{1-q}{(1+a)(1+b)}, \;  \NL
  \gamma=-\frac{(1-q)cd}{(1+c)(1+d)}, \delta=-
\frac{ab(1-q)}{(1+a)(1+b)},
\end{multline}
see \citep{Bryc-Wesolowski-2015-asep}, \cite[(74)]{essler1996representations}, \citep{uchiyama2004asymmetric}, and     \citep{Mallick-Sandow-1997}.  In this parametrization, the singularity condition
becomes  $abcd q^N=1$. %

 Since $\alpha,\beta>0$ and $\gamma,\delta\geq 0$, when
solving the resulting quadratic equations without loss of generality we  can choose $a,c>0$, and then $b,d\in(-1,0]$. The
explicit expressions are $a=\kappa_+(\beta,\delta),b=\kappa_-(\beta,\delta), c=\kappa_+(\alpha,\gamma),
d=\kappa_-(\alpha,\gamma) $, where
 \begin{equation*}\label{kappa}
\kappa_{\pm}(u,v)=\frac{1-q-u+v\pm\sqrt{(1-q-u+v)^2+4u v}}{2u}.
\end{equation*}

Recall the $q$-hypergeometric function notation
$$
{_{r+1}\phi_r}\left(\begin{matrix}
	a_1,\dots,a_{r+1}\\ b_1,\dots,b_r
\end{matrix}\middle|q;z\right) = \sum_{k=0}^\infty \frac{(a_1, a_2, \dots,a_{r+1};q)_k}{(q,b_1, b_2, \dots,b_r;q)_k}z^k.
$$
Here we use the usual Pochhammer notation:
$$(a_1, a_2, \dots,a_r;q)_n=(a_1;q)_n(a_2;q)_n \dots(a_r;q)_n$$ and $(a;q)_{n+1}=(1-a q^{n })(a;q)_{n}$ with $(a;q)_0=1$.
Later, we will also need the $q$-numbers $[n]_q=1+q+\dots+q^{n-1}$ with the convention $[0]_q=0$, $q$-factorials  $[n]_q!=[1]_q\dots [n]_q=(1-q)^{-n}(q;q)_n$ with the convention $[0]_q!=1$,
 and the $q$-binomial coefficients
$$
\Binq{n}{k}=\frac{[n]_q!}{[k]_q![n-k]_q!}.
$$

We define the $n$-th Askey-Wilson polynomial     using the $_4\phi_3$-hypergeometric function, which in the second expression we write  more explicitly for all $x$ rather than for $x=\cos \psi$.
\begin{multline}
  \label{AW}
  p_n(x;a,b,c,d|q)=a^{-n}(ab, ac, ad;q)_n{_4\varphi_3}\left(\begin{matrix}
q^{-n},q^{n-1}abcd,a e^{i\psi},a e^{-i\psi} \\
ab, ac, ad
\end{matrix}\middle|q;q\right)
\\  =
  a^{-n}(ab, ac, ad;q)_n\sum_{k=0}^n q^k \frac{(q^{-n},abcd q^{n-1};q)_k}{(q,ab,ac,ad;q)_k}\prod_{j=0}^{k-1}(1+a^2 q^{2j}-2a x q^j).
\end{multline}
Although this is not obvious from \eqref{AW}, it is known that %
 $p_n(x;a,b,c,d|q)$ is invariant under permutations of parameters $a,b,c,d$, and that the polynomial is well defined for all $a,b,c,d\in \CC$. However, in the singular case the degree of the polynomial varies with $n$ somewhat unexpectedly.
It is easy to see from the last expression in \eqref{AW} that if $abcd q^N=1 $,
then  for $0\leq n\leq N+1$ the degree of polynomial $p_n(x;a,b,c,d|q) $ is
$\min\{n,N+1-n\}$.
In particular, the degrees may decrease and hence there is no three step recursion, or a Jacobi matrix.

\arxiv{
Indeed,  $p_n(x;a,b,c,d|q)=a^{-n}(ab, ac, ad;q)_n Q_n(x)$ with %
$$Q_n(\cos \psi)=_4\!\!\varphi_3\left(\begin{matrix}
q^{-n},q^{n-N-1},a e^{i\psi},a e^{-i\psi} \\
ab, ac, ad
\end{matrix}\middle|q;q\right),$$ and $Q_n(x)=Q_{N+1-n}(x)$ for $0\leq n \leq N+1$.
}

The relation of $\varphi_0$ to Askey-Wilson polynomials is  more  conveniently expressed using a different pair of generators of algebra $\calA$.
Instead of $\mE,\mD$, we consider elements $\dd$ and $\ee$ given by
\begin{equation}\label{DE2de}
  \mD=\theta^2\mI+\theta \dd,\; \mE=\theta^2\mI + \theta\ee, \;\theta=1/\sqrt{1-q}.
\end{equation}
(Similar transformation was used by several authors, including \cite{uchiyama2004asymmetric} and \cite{Bryc-Wesolowski-2015-asep}.)

In this notation, $\calA$ is then an algebra with identity and two generators $\dd,\ee$ that satisfy relation
\begin{equation}
  \label{de-comute}
  \dd\ee-q\ee\dd=\mI.
\end{equation}
According to Theorem \ref{T-0}, functional $\varphi_0$ is defined on $\calA_N$ in the  singular case, and on all of $\calA$ in the non-singular case. We  include non-singular case  in the conclusion below by setting $N=\infty$.
The action of $\varphi_0$ on Askey-Wilson polynomials can now be described as follows.
\begin{theorem}%
\label{T3}
With $\xx=
\frac{1}{2\theta}\left(\ee+\dd\right)$, for $1\leq n< N+1$ we have
$$ \varphi_0\left[p_n(\xx\;;a ,b ,c,d\;|q)\right]=0.$$

 More generally, for   any non-zero
$t\in\CC$  let
\begin{equation}   \label{ed2x} \xx_t=
\frac{1}{2\theta}\left(\tfrac1t\ee+t\dd\right).
\end{equation}
Then
\begin{equation}
   \label{Phi0-mean} \varphi_0\left[p_n(\xx_t\;;at ,bt ,\tfrac{c }{t}, \tfrac{d}{t}\;|q)\right]=0 \mbox{ for } 1\leq n< N+1.
\end{equation}
\end{theorem}

The proof of Theorem \ref{T3} appears in Section \ref{sect:ProofT3} and is fairly involved.
It relies on evaluation of $\varphi_0$ on the  family of continuous $q$-Hermite
polynomials, on explicit formula for the  connection coefficients between the  $q$-Hermite polynomials and the Askey-Wilson polynomials which we did not find in the literature,
and to complete the proof we need some  non-obvious $q$-hypergeometric identities.
In Appendix \ref{Sect:TASEP} we discuss action of  $\varphi_0$ and $\varphi_1$ on the Askey-Wilson polynomials in the much simpler case of the Totaly Asymmetric Exclusion process where $q=0$.

\subsection{Relation to orthogonality functional for the Askey Wilson polynomials}
In the non-singular  case when  $q^n abcd\ne 1$ for all $n=0,1,\dots$,  the Askey-Wilson
polynomials $\{p_n\}_{n=0,1,\dots}$ are of increasing degrees and  satisfy the three step recursion  \citep[(1.24)]{Askey-Wilson-85}.
According to Theorem \ref{T-0} functional $\varphi_0$   is then defined on all of $\calA$ and determines stationary probabilities \eqref{MatrixSolution+} for all $L\geq 0$.  Theorem \ref{T3} implies that $\varphi_0$ is an orthogonality functional for the Askey-Wilson polynomials,
which encodes the relation between ASEP and
 Askey-Wilson polynomials that was discovered by %
 Uchiyama, Sasamoto and Wadati \citep{uchiyama2004asymmetric}.
 In particular, \eqref{Phi0-mean} corresponds to \citep[formula (6.2)]{uchiyama2004asymmetric} with $\xi=t$.

Orthogonality can be seen as follows. Theorem \ref{T3} says that
$$\varphi_0\left[p_n(\xx\;;a ,b ,c,d\;|q)\right]=0$$
 for all $n\geq 1$, and it is easy to check, see e.g. \cite[Proof of Favard's theorem]{chihara2011introduction}, that the latter property together with the three-step recursion for the Askey-Wilson polynomials implies
 orthogonality:
$$ \varphi_0\left[p_m(\xx\;;a ,b ,c,d |q)p_n(\xx\;;a ,b ,c,d\;|q)\right]=0$$ for all $m\ne n$.
This orthogonality relation holds without   additional conditions on $a,b,c,d$ that appear when orthogonality of polynomials
$\{p_n\}$ is considered on the real line \cite[Theorem 2.4]{Askey-Wilson-85}, or on a complex curve \cite[Theorem 2.3]{Askey-Wilson-85}. %
Since $\varphi_0\left[p_n(\xx\;;a ,b ,c,d\;|q)\right]\ne 0$  only for $n=0$, linearization formulas \cite{foupouagnigni2013connection}
give the value of
\begin{multline*}
\varphi_0\left[p_n^2(\xx\;;a,b,c,d\;|q)\right] = \frac {   (ab,ac,ad;q)^2_n  }{  a^{2n}   }
         \sum _ {L = 0}^{2n  } \frac { q^
        L \left( a b  ,
          a c  ,
          a d   ;q \right)_L   }
          {\left(a b c d ;q  \right)_L}
     \NL \times  \sum _ {j = \max (0,
          L - n)}^{\min (n,
         L  )} \frac {q^{j (j - L
      )} \left(q^{-n},a b c d q^{n - 1};q\right)_j }
          { (q, ab, ac, ad;q)_j  (q;q)_{  L -j}}\\
          \times \sum _ {k =
          0}^{\min (j, j - L + n)} \frac {q^
          k  \left(q^{-j},
         a^2 q^{L + r};q\right)_k  \left(q^{-n}, a b c d q^{n - 1} ;q\right)_{
           k + L   -j}} { (q)_k (a b,  a c,  a d;q)_{ k + L  -j}},
\end{multline*}
which  may fail to be positive when $abcd>1$.

\arxiv{
Somewhat more generally, in the notation of \cite{foupouagnigni2013connection} we have
$$ \varphi_0\left[p_m (\xx\;;a,b,c,d\;|q)p_n (\xx\;;a,b,c,d\;|q)\right]=L_0(m,n),$$
where
\begin{multline*}
 L_r(m,n)= \frac { q^{\frac {1} {2} r (r + 1)}(ab,ac,ad;q)_m (ab,ac,ad;q)_n  }{(-1)^
    r a^{m + n - r}   \left(a b c d q^{r - 1} ;q\right)_r}
         \sum _ {L = 0}^{m + n +
      r}\Binq{L+r}{r} \frac { q^
        L \left( a b q^r,
          a c q^r,
          a d q^r  ;q\right)_L   }
          {\left(a b c d q^{2 r} ;q \right)_L}
          \\
          \times \sum _ {j = \max (0,
          L - m + r)}^{\min (n,
         L +
          r)} \frac {q^{j (j - L -
              r)} \left(q^{-n},a b c d q^{n - 1};q\right)_j }
          { (q, ab, ac, ad;q)_j  (q;q)_{  L +
         r-j}}\\
          \times \sum _ {k =
          0}^{\min (j, j - L + m - r)} \frac {q^
          k  \left(q^{-j},
         a^2 q^{L + r};q\right)_k  \left(q^{-m}, a b c d q^{m - 1} ;q\right)_{
           k + L + r -j}} { (q)_k (a b,  a c,  a d;q)_{ k + L +
             r-j}}.
\end{multline*}
Numerical experiments suggest that $L_0(m,n)=0$ if $p_n,p_m$ have different degrees which, if true, would strengthen the conclusion of Theorem \ref{T3} to the assertion of full orthogonality.
}

\begin{remark}After this paper was submitted, we learned about  Ref. \cite{lemay2018q} which introduces nonstandard truncation condition for the Askey-Wilson polynomials in the singular case $abcdq^N=1$.
Their $q$-para-Racah polynomials are obtained by taking a limit for special choices of positive parameters $b,d$ which do not arise from ASEP.
Finite dimensional representations of the Askey-Wilson algebra in the singular case are discussed in  \cite[Section 7]{aneva2007matrix}, \cite[page 15]{Aneva-2009-Integrability} and \cite[Section 4]{tsujimoto2017tridiagonal}.
\end{remark}

\section{Proof of Theorem \ref{T-0}}\label{sect:ProofT1}
We begin with two observations from the literature. %
The first observation is  that the proof of Derrida, Evans, Hakim and Pasquier in \cite{derrida1993exact} is non-recursive, so it implies that  an invariant  functional on the finite-dimensional subspace $\calA_L$
  determines stationary probabilities for ASEP of size $L$.
\begin{theorem}[\cite{derrida1993exact}]\label{T2} Fix $L\in\NN$. Suppose that    $\varphi$ is a  linear functional  on
$\calA_L $ such that $\varphi\left[(\mE+\mD)^L\right]\ne 0$. If  invariance equations \eqref{W+} and \eqref{V+}
hold for all  $\mA\in\calA_{L-1}$, then the stationary probabilities for the ASEP of length $L$ are
\begin{equation}
  \label{MatrixSolution+}
 P(\tau_1,\dots,\tau_L)= \frac{\varphi\left[\prod_{j=1}^L (\tau_j\mD+(1-\tau_j)\mE)\right]}{\varphi\left[(\mD+ \mE)^L\right]}.
\end{equation}
\end{theorem}
 \begin{proof} The argument here is the same as the proof in \cite[Section 11.1]{derrida1993exact}  for the matrix version, see also \cite[Section III]{sandow1994partially}. The important aspect of that proof is that it works with fixed $L$, i.e.,
 that we do not need to use a recurrence that lowers the value of $L$ as in  \cite[formula (8)]{derrida1992exact} or in  \cite[Theorem 3.2]{Liggett-1975}.
We  reproduce  a version of argument from \citep{derrida1993exact} for completeness and clarity.

For $L=1$ it is easily seen that the stationary distribution is
 $P(1)=\frac{\alpha+\delta}{\alpha+\beta+\gamma+\delta}$ with $P(0)=1-P(1)$. On the other hand, equations \eqref{W+} and \eqref{V+} give
 $\alpha \varphi[\mE]-\gamma \varphi[\mD]=\varphi[\mI]$ and $\beta \varphi[\mD]-\delta  \varphi[\mE]=\varphi[\mI]$.
The solution is:
 $$\varphi[\mE]=
 \begin{cases}
\frac{\beta+\gamma}{\alpha \beta -\gamma\delta} \varphi[\mI] & \mbox{ if } \alpha\beta \ne \gamma\delta \\
\frac{\gamma}{\alpha+\gamma} & \mbox{ if } \alpha\beta = \gamma\delta
 \end{cases}, \quad \;
\varphi[\mD]=
 \begin{cases}
\frac{\alpha+\delta}{\alpha \beta -\gamma\delta}\varphi[\mI] & \mbox{ if } \alpha\beta \ne \gamma\delta \\
\frac{\alpha}{\alpha+\gamma} & \mbox{ if } \alpha\beta = \gamma\delta
 \end{cases},$$
 where we  note that $\varphi[\mI]=0$ when $\alpha\beta=\gamma\delta$ and in this case we also used the normalization $\varphi[\mE+\mD]=1$ to determine the values.
In both cases, a calculation shows that
$$\frac{\varphi[\mD]}{\varphi[\mE]+\varphi[\mD]}=\frac{\alpha+\delta}{\alpha+\beta+\gamma+\delta}$$
giving the correct value of $P(1)$.

Suppose  that $L\geq 2$. Denote by $p(\tau_1,\dots,\tau_L)=  \varphi\left[\prod_{j=1}^L (\tau_j\mD+(1-\tau_j)\mE)\right]$
the un-normalized probabilities.  Since by assumption the denominator in \eqref{MatrixSolution+} is non-zero, it is enough to
verify that the right hand side of \eqref{DiffEq} vanishes on $p(\tau_1,\dots,\tau_L)$.
That is, we want to show that
\begin{multline}\label{MS2}
(\delta_{\tau_1=1}-\delta_{\tau_1=0})\left[\alpha p(0,\tau_2,\dots,\tau_L)-\gamma  p(1,\tau_2,\dots,\tau_L)\right] \\
+\sum_{k=1}^{L-1}
(\delta_{\tau_k=0,\tau_{k+1}=1}-\delta_{\tau_k=1,\tau_{k+1}=0})\Big[ p(\tau_1,\dots,\tau_{k-1},1,0,\tau_{k+2},\dots,\tau_L)
\\ -qp(\tau_1,\dots,\tau_{k-1},0,1,\tau_{k+2},\dots,\tau_L)
\Big] \NL
+(\delta_{\tau_L=0}-\delta_{\tau_L=1})\left[\beta p(\tau_1,\dots,\tau_{L-1},1)-\delta p(\tau_1,\dots,\tau_{L-1},0)\right] =0.
\end{multline}
Denote
\begin{equation*}
  \label{XY}
 \mX_k=\prod_{j=1}^{k} (\tau_j\mD+(1-\tau_j)\mE)  \mbox{ and } \mY_k=\prod_{j=k}^L (\tau_j\mD+(1-\tau_j)\mE)
\end{equation*}
with the usual convention that empty products are $\mI$. Relation \eqref{q-comm-Derrida} implies that \begin{multline*}
p(\tau_1,\dots,\tau_{k-1},1,0,\tau_{k+2},\dots,\tau_L)-qp(\tau_1,\dots,\tau_{k-1},0,1,\tau_{k+2},\dots,\tau_L) \\=
\varphi[\mX_{k-1}(\mD\mE-q\mE\mD)\mY_{k+2}]=\varphi[\mX_{k-1}(\mD+\mE)\mY_{k+2}].
\end{multline*}
Noting that
$$\delta_{\tau_k=0,\tau_{k+1}=1}-\delta_{\tau_k=1,\tau_{k+1}=0}=(1-\tau_k)\tau_{k+1}-\tau_k(1-\tau_{k+1})=\tau_{k+1}-\tau_k,$$
the sum in \eqref{MS2} becomes
$$\sum_{k=1}^{L-1} (\tau_{k+1}-\tau_k)\varphi[\mX_{k}(\mD+\mE)\mY_{k+2}].$$
Since $\tau_k,\tau_{k+1}\in\{0,1\}$, the difference $\tau_{k+1}-\tau_k$ can take only three values $0,\pm 1$.
Considering all 
possible cases, we get
\begin{multline*}
 (\tau_{k+1}-\tau_k)\varphi[\mX_{k-1}(\mD+\mE)\mY_{k+2}]
\NL =(\tau_{k+1}-\tau_k)\Big(\varphi[\mX_{k-1}(\tau_{k+1}\mD+(1-\tau_{k+1})\mE)\mY_{k+2}]
\\ +
 \varphi[\mX_{k-1}(\tau_{k}\mD+(1-\tau_{k})\mE)\mY_{k+2}]\Big) \NL
 =(\tau_{k+1}-\tau_k)\left(\varphi[\mX_{k} \mY_{k+2}]+\varphi[\mX_{k-1} \mY_{k+1}]\right)
 \\
 =\eps_k \varphi[\mX_{k-1} \mY_{k+1}]-\eps_{k+1}\varphi[\mX_{k} \mY_{k+2}],
\end{multline*}
where 
 $\eps_k=\delta_{\tau_k=0}-\delta_{\tau_k=1}=\pm1$. 
 (For the last equality we need to notice that $ \mX_{k-1}
\mY_{k+1}=\mX_{k} \mY_{k+2}$ when $\tau_k=\tau_{k+1}$.)
\arxiv{Indeed,
\begin{equation}\label{*}
   \mX_{k-1}\mY_{k+1} - \mX_{k}\mY_{k+2}=\mX_{k-1}((\tau_{k+1}-\tau_k)\mD+ (\tau_k-\tau_{k+1})\mE)\mY_{k+2}. 
\end{equation}
The cases to consider are:
\begin{enumerate}
 \item $\tau_k=\tau_{k+1}$. Then $\eps_k=\eps_{k+1}$ and $\varphi[\mX_{k-1} \mY_{k+1}]=\varphi[\mX_{k} \mY_{k+2}]$ by \eqref{*}.
 \item $\tau_k=1$, $\tau_{k+1}=0$. Then  $\eps_k=-1$ and $\eps_{k+1}=1$.

 We have $(\tau_{k+1}-\tau_k)\left(\varphi[\mX_{k} \mY_{k+2}]+\varphi[\mX_{k-1} \mY_{k+1}]\right) =
 - \varphi[\mX_{k} \mY_{k+2}]-\varphi[\mX_{k-1} \mY_{k+1}]=\eps_k \varphi[\mX_{k-1} \mY_{k+1}]-\eps_{k+1}\varphi[\mX_{k} \mY_{k+2}]
 $ as required.
 \item $\tau_k=0$, $\tau_{k+1}=1$. Then   $\eps_k=1$ and $\eps_{k+1}=-1$.

  We have $(\tau_{k+1}-\tau_k)\left(\varphi[\mX_{k} \mY_{k+2}]+\varphi[\mX_{k-1} \mY_{k+1}]\right) =
  \varphi[\mX_{k} \mY_{k+2}]+\varphi[\mX_{k-1} \mY_{k+1}]=\eps_k \varphi[\mX_{k-1} \mY_{k+1}]-\eps_{k+1}\varphi[\mX_{k} \mY_{k+2}]
 $ as required.
\end{enumerate}
}
Thus
\begin{multline*}
   \sum_{k=1}^{L-1} (\tau_{k+1}-\tau_k)\varphi[\mX_{k}(\mD+\mE)\mY_{k+2}]=\sum_{k=1}^{L-1}(\eps_k \varphi[\mX_{k-1} \mY_{k+1}]-\eps_{k+1}\varphi[\mX_{k} \mY_{k+2}])
\NL =\eps_1\varphi[\mY_2]-\eps_L\varphi[\mX_{L-1}].
\end{multline*}
By invariance we have
$$\left[\alpha p(0,\tau_2,\dots,\tau_L)-\gamma  p(1,\tau_2,\dots,\tau_L)\right] =\varphi[(\alpha\mE-\gamma\mD)\mY_2]=\varphi[\mY_2]
$$
$$\left[\beta p(\tau_1,\dots,\tau_{L-1},1)-\delta p(\tau_1,\dots,\tau_{L-1},0)\right]=\varphi[\mX_{L-1}(\beta\mD-\delta\mE)]=\varphi[\mX_{L-1}].$$
So the left hand side of \eqref{MS2}  becomes
$$-\eps_1 \varphi[\mY_2]+\eps_1\varphi[\mY_2]-\eps_L\varphi[\mX_{L-1}]+\eps_L\varphi[\mX_{L-1}]=0$$
proving \eqref{MS2}.
\end{proof}

The second observation is that stationary distribution for ASEP of length $L=N+1$ is given as an explicit product of Bernoulli measures.
This  fact has been explicitly
  noted in \cite[Section 5.2]{enaud2004large}, see also \cite[Section 4.6.2]{Enaud2005} and \cite[Section 3]{Aneva-2009-Integrability}. The proof consists of verification of detailed balance equations so that individual terms on the right hand side of \eqref{DiffEq} vanish.

\begin{proposition}[\citet{enaud2004large}]\label{P1}
  Suppose $\alpha\beta=q^N\gamma\delta$.   If $L=N+1$   then the stationary  distribution of the ASEP is the product of Bernoulli measures
    $$P(\tau_1,\dots,\tau_L)= \prod_{j=1}^L p_j^{\tau_j} q_j^{1-\tau_j}$$
with $p_j=\frac{\alpha}{\alpha+\gamma q^{j-1}}$ and $q_j=1-p_j$.

\end{proposition}
\arxiv{
\begin{proof}
The stationary distribution for $L=1$ is $p_1=\frac{\alpha+\delta}{\alpha+\beta+\gamma+\delta}$.  When  $\alpha\beta=\gamma\delta$
 this answer matches $p_1=\frac{\alpha}{\alpha+\gamma }$.

 For $L\geq 2$ we can use \eqref{DiffEq}.  Inserting the product measure  into the right hand side of \eqref{DiffEq}, we get:
  $$
\alpha P_t(0,\tau_2,\dots,\tau_L)-\gamma  P_t(1,\tau_2,\dots,\tau_L)=\alpha\frac{\gamma}{\alpha+\gamma}\prod_{k>1}p_{k}^{\tau_k}q_k^{1-\tau_k}-\gamma \frac{\alpha}{\alpha+\gamma} \prod_{k>1}p_{k}^{\tau_k}q_k^{1-\tau_k}=0,
  $$
 \begin{multline*}
\left[P_t(\tau_1,\dots,\tau_{k-1},1,0,\tau_{k+2}\dots)- q P_t(\tau_1,\dots,\tau_{k-1},0,1,\tau_{k+2}\dots)
\right] \\ =\prod_{i\leq k-1}p_{i}^{\tau_i}q_i^{1-\tau_i}\frac{\alpha}{\alpha+q^{k-1}\gamma}\frac{\gamma q^k}{\alpha+q^{k}\gamma}\prod_{j\geq k+2}p_{j}^{\tau_j}q_j^{1-\tau_j}
\\
-q\prod_{i\leq k-1}p_{i}^{\tau_i}q_i^{1-\tau_i}\frac{\gamma q^{k-1}}{\alpha+q^{k-1}\gamma}\frac{\alpha}{\alpha+q^{k}\gamma}\prod_{j\geq k+2}p_{j}^{\tau_j}q_j^{1-\tau_j}
=0.
\end{multline*}
 Finally,
 \begin{multline*}\left[\beta P_t(\tau_1,\dots,\tau_{L-1},1)-\delta P_t(\tau_1,\dots,\tau_{L-1},0)\right]
 =
\prod_{i\leq L-1}p_{i}^{\tau_i}q_i^{1-\tau_i}\left(\beta\frac{ \alpha}{\alpha+q^{L-1} \gamma} -\delta \frac{ \gamma q^{L-1}}{\alpha+q^{L-1} \gamma}\right)
=0,
\end{multline*}
as $L=N+1$ and $\alpha\beta=q^N\gamma\delta$. This shows that the right hand side of \eqref{DiffEq} is zero, i.e. the product measure is stationary.
\end{proof}}

 \subsection{Construction of the pair of invariant functionals}

The construction starts with choosing a convenient basis for $\calA$, consisting of monomials in   normal order, with all factors $\ee$ occurring before $\dd$.
Such monomials appear in many
references, see e.g. \citet[pg 368]{frisch1970parastochastics}, \citet[page 137]{bozejko97qGaussian}, \citet[page 4524]{Mallick-Sandow-1997}, or \cite[Eq. (19)]{derrida-mallick-97}.
\begin{proposition}
Monomials in normal order $\{\ee^m\dd^n: m,n=0,1,\dots\}$ are a basis of $\calA$ considered as a vector space. In this basis $\calA_k$ is the span of $\{\ee^m\dd^n: m+n\leq k\}$.
\end{proposition}
\begin{proof}
It is easy to check by induction that $q$-commutation relation  \eqref{de-comute}
gives explicit expressions for ``swaps" that recursively convert all monomials into linear combinations of monomials in normal order. We have
   \begin{equation}\label{de^m}
   \dd\ee^m\dd^n=q^m\ee^m \dd^{n+1}+[m]_q\ee^{m-1}\dd^n.
\end{equation}
\arxiv{ Indeed,   $\dd\ee^m =q^m\ee^m \dd +[m]_q\ee^{m-1}$  holds for $m=0,1$.  For the induction step we use \eqref{de-comute} and get
$\dd\ee^{m+1} =q^m\ee^m \dd\ee +[m]_q\ee^{m}=q^m\ee^m (q\ee\dd+\mI) +[m]_q\ee^{m}=q^{m+1}\ee^{m+1}\dd+(q^m+[m]_q)\ee^m=q^{m+1}\ee^{m+1}\dd+[m+1]_q\ee^m$. To get the general case of \eqref{de^m} we just right-multiply the formula $\dd\ee^m =q^m\ee^m \dd +[m]_q\ee^{m-1}$ by $\dd^n$.
}
 Similarly, we get
   \begin{equation}\label{d^ne}
 \ee^m\dd^n\ee=q^n\ee^{m+1} \dd^{n}+[n]_q\ee^{m}\dd^{n-1}.
\end{equation}
\arxiv{As before, we only need to prove $\dd^n\ee=q^n\ee\dd+[n]_q\dd^{n-1}$. The induction step is
$\dd^{n+1}\ee=\dd^n(\dd\ee)=\dd^n (q\ee \dd+\mI)= q^{n+1}\ee \dd^{n+1}+(q[n]_q+1)\dd^n= q^{n+1}\ee \dd^{n+1}+[n+1]_q\dd^n$.
}
(Formulas \eqref{de^m} and \eqref{d^ne} holds also for $m=0$ or $n=0$ after omitting the  term with $[0]_q=0$.)

The formulas imply that
any monomial is a linear combination of monomials in normal order:
\begin{equation}\label{Ordered0}
\dd^{n_1}\ee^{m_1}\dots \dd^{n_k}\ee^{m_k}=q^I\ee^{m}\dd^n+\sum_{i+j\leq m+n-1}a_{i,j}\ee^i\dd^j,
\end{equation}
where $m=m_1+\dots m_k$, $n=n_1+\dots+n_k$ and $I=\sum_{i=1}^k\sum_{j=1}^i m_i n_j$ is the minimal number of inversions (length) of a permutation that maps
$\ee^{m}\dd^n$ into $\dd^{n_1}\ee^{m_1}\dots \dd^{n_k}\ee^{m_k}$, see e.g. \cite{bjorner2006combinatorics}.
Compare \cite[Appendix A]{Mallick-Sandow-1997}.

Formula \eqref{Ordered0} shows that monomials in normal order span $\calA$. To verify that they are linearly independent  we consider  a pair of
linear mappings (endomorphism)
   $\Dq$ and $\Z$ acting on  polynomials
$\CC[z]$  which are the $q$-derivative and the multiplication mappings:
 $$(\Dq p)(z)=\frac{p(z)-p(qz)}{(1-q)z},\;  (\Z p)(z)=zp(z).$$ The mapping $\dd\mapsto\Dq$ and $\ee\mapsto\Z$
extends to homomorphism of   algebra $\CC\langle\dd,\ee\rangle$ of polynomials in noncommuting variables $\ee,\dd$ to the algebra $End(\CC[z])$.
 It is well known that
$\Dq\Z-q\Z\Dq$ is the identity, so we get an induced homomorphism  of algebras
$$\calA=\CC\langle\dd,\ee \rangle/\mathcal{I}\to End(\CC[z]),$$
where $\mathcal{I}$ is the two sided ideal generated by $\dd\ee-q\ee\dd-\mI$.
Therefore, it is enough to prove linear independence of  $\{\Z^m\Dq^n\}$. %

To prove the latter, consider a finite sum $\SSS=\sum_{m,n\geq 0} a_{m,n}\Z^m\Dq^n =0$ and suppose that some of the coefficients $a_{m,n}$ are non-zero. Let $n_*\geq 0$ be the
smallest value of index $n$ among the non-zero coefficient $a_{m,n}$.
 We note that
\begin{equation*}
  \Z^m\Dq^n (z^{n_*})=\begin{cases}
    0, & n>n_* \\
    [n_*]_q! z^{m}, &n=n_*
  \end{cases}
\end{equation*}
Therefore, applying   $\SSS$ to the monomial $z^{n_*}\in\CC[z]$ we get
$$\sum_{m\in M} a_{m,n_*} [n_*]_q! z^{m} =0,$$ i.e., all
$\{a_{m,n_*}: m\in M\}$ are zero, in contradiction to our choice of $n_*$. The contradiction shows that all coefficients must be zero, proving linear independence.
\end{proof}

Using \eqref{DE2de} we remark that invariance conditions  \eqref{W+} and \eqref{V+} with $\mA\in\calA_k$
 can be written equivalently in our
basis of monomials in normal order as
\begin{eqnarray}
  \label{W++}  \alpha\varphi[\ee^{m+1}\dd^n]-\gamma \varphi[\dd\ee^m\dd^n]&=&\Delta(\gamma-\alpha)\varphi[\ee^{m}\dd^n], \\
 \label{V++}  -\delta\varphi[\ee^{m}\dd^n\ee]+\beta \varphi[\ee^m\dd^{n+1}]&=&\Delta(\delta-\beta)\varphi[\ee^{m}\dd^n],
\end{eqnarray}
where $m+n\leq k$ and  %
$\Delta(x)=\theta^{-1}+\theta x $.
\subsection{Recursive construction of  the functionals}
We define linear functional $\varphi=\varphi_0$ or $\varphi=\varphi_1$ by assigning its  values on all elements of the basis
$\{\ee^m\dd^n\}$
 and then extending  it to  $\calA_N$ or  $\calA$ by linearity. On the basis, we  define $\varphi$ recursively, extending it from $\calA_k$ to $\calA_{k+1}$
in such a way that the invariance properties \eqref{W+} and \eqref{V+} hold.

\subsubsection{Initial values}
  We set $\varphi_0[\mI]=1$.
   We set
\begin{equation}
  \label{phi-1-ini}
  \varphi_1[\ee^{m}\dd^n]= \begin{cases}
0  & \mbox{if  $m+n\leq N$}\\
\Pi^{-1}\alpha^n\gamma^m q^{m(m-1)/2} & \mbox{if  $m+n= N+1$},
\end{cases}
\end{equation}
where the normalizing constant $\Pi=\theta^{N+1}\prod_{j=1}^{N+1}(\alpha+q^{j-1}\gamma)$ is chosen so that
$\varphi_1\left[(\ee+\dd)^{N+1}\right]=1/\theta^{N+1}$.

Clearly, $\varphi_1\equiv 0$ on $\calA_N$.
We  need to check that our initialization of $\varphi_1$  has the properties we need for the recursive construction:
that invariance  conditions hold  for $\mA\in\calA_N$, and that  $\varphi_1$  determines the stationary measure of ASEP with $L=N+1$.

\begin{lemma}\label{L-phi12prod}    For monomials of degree $N+1$ we have
\begin{equation}
  \label{phi-prod}
  \varphi_1[\mD^{\tau_1}\mE^{1-\tau_1}\dots \mD^{\tau_{N+1}}\mE^{1-\tau_{N+1}}]=\prod_{j=1}^{N+1}p_j^{\tau_j}q_j^{1-\tau_j},
\end{equation}
where the weights $\{p_j\}$ come from stationary product measure in  Proposition \ref{P1}. Furthermore, \eqref{W+} and \eqref{V+} hold for $\mA\in\calA_N$.
\end{lemma}

\begin{proof}
Since $\varphi_1$ vanishes on polynomials of lower degree, from \eqref{DE2de} it is easy to see that
 $$\varphi_1[\mD^{\tau_1}\mE^{1-\tau_1}\dots \mD^{\tau_{N+1}}\mE^{1-\tau_{N+1}}]=\theta^{N+1}
 \varphi_1[\dd^{\tau_1}\ee^{1-\tau_1}\dots \dd^{\tau_{N+1}}\ee^{1-\tau_{N+1}}].$$
So we only need to show that
\begin{equation}
  \label{swap} \varphi_1[\dd^{\tau_1}\ee^{1-\tau_1}\dots
\dd^{\tau_{N+1}}\ee^{1-\tau_{N+1}}]=\prod_{j=1}^{N+1}p_j^{\tau_j}q_j^{1-\tau_j}/\theta^{N+1}.
\end{equation}
It is easy to see that this formula holds true for $\varphi_1[\ee^m\dd^{N+1-m}]$. (In fact, this is how we defined $\varphi_1[\ee^m\dd^n]$ when $m+n=N+1$.)
 All  monomials of the form $\dd^{\tau_1}\ee^{1-\tau_1}\dots
\dd^{\tau_{N+1}}\ee^{1-\tau_{N+1}}$ can be obtained from monomials $\ee^m\dd^{N+1-m}$ in normal order by applying a finite number of
 adjacent transpositions, i.e., by swapping  pairs of  adjacent factors $\ee\dd$ or $\dd\ee$. (Adjacent transpositions are Coxeter generators
 for the permutation group, see e.g.  \cite{bjorner2006combinatorics}.)
So to  complete the proof we check that if  formula \eqref{swap} holds for some monomial, then it also holds after we swap the entries at adjacent locations $k,k+1$.
Suppose that
$$\theta^{N+1}\varphi_1[\mX\ee \dd \mY]=q_kp_{k+1}\Pi'=
\frac{\alpha \gamma q^{k-1}}{(\alpha+q^{k-1}\gamma)(\alpha+q^k\gamma)} \Pi',$$
with $\mX=\dd^{\tau_1}\ee^{1-\tau_1}\dots \dd^{\tau_{k-1}}\ee^{1-\tau_{k-1}}$,  $\mY=\ee^{1-\tau_{k+2}}\dots
\dd^{\tau_{N+1}}\ee^{1-\tau_{N+1}}$ and $\Pi'=\prod_{  j\ne k,k+1} p_j^{\tau_j}q_j^{1-\tau_j}$.
Multiplying this by $q$ and replacing  $q\ee\dd$ by $ \dd\ee-\mI$, we get
$$
\theta^{N+1}\varphi_1[\mX\dd \ee\mY]=
\frac{\alpha \gamma q^{k}}{(\alpha+q^{k-1}\gamma)(\alpha+q^k\gamma)} \Pi'=p_kq_{k+1}\Pi',$$
as $\varphi_1$ vanishes on lower degree monomials. So   the swap   preserves the expression
on the right hand side of \eqref{swap}. The case  when the factors at the adjacent locations are $\dd\ee$ is handled similarly.

 To verify that \eqref{W+} and \eqref{V+} hold for $\mA\in\calA_N$ we show that \eqref{W++} and
\eqref{V++} hold for $m+n\leq N$. Indeed, both sides are zero if $m+n\leq N-1$, and if $m+n=N$ then the right hand sides are
still zero. By \eqref{swap}, the left hand side of \eqref{W++} is
$$%
 \alpha^{n+1}\gamma^{m+1} \left(q^{m(m+1)/2}-q^m   q^{m(m-1)/2}\right)/\Pi=0.$$
The left hand side of \eqref{V++} is
$$
\alpha^{n}\gamma^{m} q^{m(m-1)/2}\left(\alpha\beta-q^{n+m}\gamma\delta \right)/\Pi=0$$
by singularity assumption.
\end{proof}

\subsubsection{Recursive step for $\varphi=\varphi_0$ or $\varphi_1$} Suppose $\varphi$ is defined on $\calA_k$ and that
invariance conditions hold for $\mA\in\calA_{k-1}$. If $m+n=k$ with $1\leq k<N$ (case of $\varphi_0$)  or $k\geq N+1$
(case of $\varphi_1$). Define
\begin{multline}
  \label{SolA} \varphi[\ee^{m+1}\dd^n]=
\frac{1}{(q^N-q^{m+n})\gamma\delta}\Big[
\left(\beta\Delta(\gamma-\alpha)+\gamma\Delta(\delta-\beta)q^m\right)\varphi[\ee^m\dd^n]
\\ +\gamma\delta [n]_qq^m\varphi[\ee^m\dd^{n-1}]+
\beta\gamma[m]_q\varphi[\ee^{m-1}\dd^n] \Big],
\end{multline}
 \begin{multline}  \label{SolB} \varphi[\ee^{m}\dd^{n+1}]
= \frac{1}{(q^N-q^{m+n})\gamma\delta}\Big[ \left(\alpha \Delta(\delta-\beta)+\delta \Delta(\gamma-\alpha)q^n\right)\varphi[\ee^m\dd^n]
\\+
  \alpha\delta[n]_q\varphi[\ee^m\dd^{n-1}]
  +\gamma\delta q^n[m]_q\varphi[\ee^{m-1}\dd^n]\Big],
\end{multline}
where $\Delta(x)=\theta^{-1}+\theta x$ comes from \eqref{W++} and \eqref{V++}.

\begin{remark}\label{Rem:nonsingular}
  If $\alpha\beta-q^n\gamma\delta \ne 0$ for all $n$, we  define  $\varphi_0$ on $\calA$, replacing the above recursion with
\begin{multline}
  \label{SolA0} \varphi_0[\ee^{m+1}\dd^n]=
  \frac{1}{\alpha\beta-q^{m+n} \gamma\delta}\Big[
\left(\beta\Delta(\gamma-\alpha)+\gamma\Delta(\delta-\beta)q^m\right)\varphi_0[\ee^m\dd^n]
\\+\gamma\delta [n]_qq^m\varphi_0[\ee^m\dd^{n-1}]+
\beta\gamma[m]_q\varphi_0[\ee^{m-1}\dd^n]\Big],
\end{multline}
 \begin{multline}  \label{SolB0} \varphi_0[\ee^{m}\dd^{n+1}]=
   \frac{1}{\alpha\beta-q^{m+n} \gamma\delta}\Big[\left(\alpha \Delta(\delta-\beta)+\delta \Delta(\gamma-\alpha)q^n\right)\varphi_0[\ee^m\dd^n]
   \\   +  \alpha\delta[n]_q\varphi_0[\ee^m\dd^{n-1}]+\gamma\delta q^n[m]_q\varphi_0[\ee^{m-1}\dd^n]\Big].
\end{multline}
\end{remark}

We need to make sure that this expression is well defined.
 \begin{lemma} Fix $k\ne N$. %
Suppose  $m'+n'=k+1$. Then $\varphi[\ee^{m'}\dd^{n'}]$ is well defined:  both formulas give the same answer  when $(m',n')$ can be represented   as
  $(m',n')=(m+1,n)$ and as $(m',n')=(m,n+1)$.
 \end{lemma}

\begin{proof}

We proceed by contradiction. Suppose that $m,n$ is a pair of smallest degree $m+n$ where consistency fails.
This means that \eqref{W++} and \eqref{V++}  still hold for all pairs of lower degree but the solution \eqref{SolB} with
$m$ replaced by $m+1$ and $n$ replaced by $n-1$
does not match the solution in  \eqref{SolA}. We show that this cannot be true by verifying that the numerators are the same,
\begin{multline}
  \label{consistency}
  \left(\beta\Delta(\gamma-\alpha)+\gamma\Delta(\delta-\beta)q^m\right)\varphi[\ee^m\dd^n]
  \NL +\gamma\delta [n]_qq^m\varphi[\ee^m\dd^{n-1}]+\beta\gamma[m]_q\varphi[\ee^{m-1}\dd^n]
  \\=\left(\alpha \Delta(\delta-\beta)+\delta \Delta(\gamma-\alpha)q^{n-1}\right)\varphi[\ee^{m+1}\dd^{n-1}]
  \\+ \alpha\delta[n-1]_q\varphi[\ee^{m+1}\dd^{n-2}]+\gamma\delta q^{n-1}[m+1]_q\varphi[\ee^{m}\dd^{n-1}].
\end{multline}
(Formally, the term with the factor $[n-1]_q$ should be omitted when $n=1$.)
The difference between the left hand side and the right hand side of \eqref{consistency} is
\begin{multline*}
   \Delta(\gamma-\alpha)\left(\beta\varphi[\ee^m\dd^n]-\delta q^{n-1}\varphi[\ee^{m+1}\dd^{n-1}]\right)
\NL +\Delta(\delta-\beta)\left(\gamma q^m\varphi[\ee^m\dd^n]-\alpha\varphi[\ee^{m+1}\dd^{n-1}]\right)\\+
\left(\gamma\delta [n]_qq^m\varphi[\ee^m\dd^{n-1}]-\alpha\delta[n-1]_q
\varphi[\ee^{m+1}\dd^{n-2}]\right)
\NL+ \left(\beta\gamma[m]_q\varphi[\ee^{m-1}\dd^n]-\delta\gamma q^{n-1}[m+1]_q\varphi[\ee^{m}\dd^{n-1}]\right).
\end{multline*}
Since $q^m[n]_q=q^m[n-1]_q+q^{m}q^{n-1}$ and $q^{n-1}[m+1]_q=q^{n-1}[m]_q+ q^{m}q^{n-1}$,  canceling the   terms with factor $q^{m}q^{n-1}$  we rewrite the above as
\begin{multline*}
   \Delta(\gamma-\alpha)\left(\beta\varphi[\ee^m\dd^n]-\delta q^{n-1}\varphi[\ee^{m+1}\dd^{n-1}]\right)
\NL+\Delta(\delta-\beta)\left(\gamma q^m\varphi[\ee^m\dd^n]-\alpha\varphi[\ee^{m+1}\dd^{n-1}]\right)\\+
\delta[n-1]_q\left(\gamma q^m\varphi[\ee^m\dd^{n-1}]-\alpha
\varphi[\ee^{m+1}\dd^{n-2}]\right)
\NL+ \gamma[m]_q\left(\beta\varphi[\ee^{m-1}\dd^n]-\delta q^{n-1}\varphi[\ee^{m}\dd^{n-1}]\right).
\end{multline*}
 We now use \eqref{de^m} and \eqref{d^ne}. We get
\begin{multline*}
   \Delta(\gamma-\alpha)\left(\beta\varphi[\ee^m\dd^n]-\delta  \varphi[\ee^{m}\dd^{n-1}\ee]\right)+\Delta(\gamma-\alpha) \delta[n-1]_q\varphi[\ee^m\dd^{n-2}]
   \\
+\Delta(\delta-\beta)\left(\gamma \varphi[\dd\ee^m\dd^{n-1}]-\alpha\varphi[\ee^{m+1}\dd^{n-1}]\right)-\Delta(\delta-\beta)\gamma[m]_q\varphi[\ee^{m-1}\dd^{n-1}]\\+
\delta[n-1]_q\left(\gamma\varphi[\dd\ee^m\dd^{n-2}]-\alpha
\varphi[\ee^{m+1}\dd^{n-2}]\right)-\gamma\delta[n-1]_q [m]_q\varphi[\ee^{m-1}\dd^{n-2}]
\\+ \gamma[m ]_q\left(\beta\varphi[\ee^{m-1}\dd^n]-\delta  \varphi[\ee^{m-1}\dd^{n-1}\ee]\right)+\gamma\delta[m]_q[n-1]_q\varphi[\ee^{m-1}\dd^{n-2}].
\end{multline*}
After canceling $\gamma\delta[m]_q[n-1]_q\varphi[\ee^{m-1}\dd^{n-2}]$ we re-group the expression into  the sum $S_1+S_2+S_3$ with
\begin{equation*}
S_1=   \Delta(\gamma-\alpha)\left(\beta\varphi[\ee^m\dd^n]-\delta  \varphi[\ee^{m}\dd^{n-1}\ee]\right)
\NL -
   \Delta(\delta-\beta)\left(\alpha\varphi[\ee^{m+1}\dd^{n-1}]-\gamma \varphi[\dd\ee^m\dd^{n-1}]\right) ,
\end{equation*}
$$S_2=  \delta[n-1]_q\left[\Delta(\gamma-\alpha) \varphi[\ee^m\dd^{n-2}]
 -
\left(\alpha
\varphi[\ee^{m+1}\dd^{n-2}]-\gamma\varphi[\dd\ee^m\dd^{n-2}]\right)\right],
   $$
$$   S_3=
\gamma[m]_q\left[ \left(\beta\varphi[\ee^{m-1}\dd^n]-\delta  \varphi[\ee^{m-1}\dd^{n-1}\ee]\right)-\Delta(\delta-\beta)\varphi[\ee^{m-1}\dd^{n-1}] \right]
.$$
From \eqref{W++} and \eqref{V++} we see that $S_1,S_2,S_3$ are zero, proving \eqref{consistency}.
\end{proof}

Formulas \eqref{SolA} and \eqref{SolB} extend $\varphi$ from  $\calA_{k}$ to $\calA_{k+1}$.
 \begin{lemma}
   Invariance  conditions \eqref{W+} and \eqref{V+} hold for $\mA\in\calA_{k}$.
 \end{lemma}

\begin{proof} We verify \eqref{W++} and \eqref{V++} with $m+n\leq k$. By inductive assumption \eqref{W++} and \eqref{V++}  hold  when $m+n< k$, so we only need to consider $m+n=k$.

Using ``swap identities" \eqref{de^m} and \eqref{d^ne} we rewrite these relations as
\begin{equation}
  \label{W+++}  \alpha\varphi[\ee^{m+1}\dd^n]-q^m\gamma \varphi[\ee^m\dd^{n+1}]=\Delta(\gamma-\alpha)\varphi[\ee^{m}\dd^n]
  + \gamma[m]_q\varphi[\ee^{m-1}\dd^{n}]
\end{equation}
and
\begin{equation}
 \label{V+++}  -q^n\delta\varphi[\ee^{m+1}\dd^n]+\beta \varphi[\ee^m\dd^{n+1}]=\Delta(\delta-\beta)\varphi[\ee^{m}\dd^n]
 \NL +\delta [n]_q\varphi[\ee^m\dd^{n-1}],
\end{equation}
with the solution given in  \eqref{SolA} and \eqref{SolB}. By linearity this establishes invariance conditions for  all $\mA\in\calA_k$.
\end{proof}

 \subsection{Signs of $\varphi$ on monomials}
To verify that $\varphi[(\mE+\mD)^L]\ne 0$, we will   need the following version of a formula discussed in \cite[Appendix A]{Mallick-Sandow-1997}.
\begin{lemma}
If
$\mX= \mE^{m_1}\dots \mD^{n_k}\mE^{m_k}\mD^{n_k}$ is a monomial  of degree $m+n$ with $m=m_1+\dots+m_k$, $n=n_1+ \dots+n_k$, then there exist non-negative integers $b_j,c_j$ and monomials $\mY_j,\mZ_j $ of degree $m+n$ such that
\begin{equation}\label{ED-expansion}
  \mX\mE =q^n \mE\mX+\sum_j b_j \mY_j \mbox{ and }  \mD\mX  =q^m  \mX\mD+\sum_j c_j \mZ_j.
\end{equation}
\end{lemma}
\begin{proof}
   Denote $\mS=\mE+\mD$. Suppose that formulas hold for $\mX$ with $k\geq 0$ factors. Then for $n=n_{k+1}$ and $m=m_{0}$ by repeated applications of \eqref{q-comm-Derrida} we get
\begin{equation}
   \label{ED-expansion1}
   \mD^n\mE=q\mD^{n-1}\mE\mD+\mD^{n-1}\mS=q^2\mD^{n-2}\mE\mD^2+\mD^{n-2}\mS\mD+\mD^{n-1}\mS=
   \dots
\NL  =q^n\mE\mD^n+\sum_{j=0}^{n-1} \mD^{n-1-j}\mS\mD^{j}
\end{equation}
   and
\begin{equation}
   \label{ED-expansion2}
      \mD\mE^m=q \mE \mD\mE^{m-1}+\mS\mE^{m-1}=q^2 \mE^2 \mD\mE^{m-2}+\mE\mS\mE^{m-2}+\mS\mE^{m-1}=\dots
    \NL  =q^m\mE^m\mD+\sum_{j=0}^{m-1} \mE^j\mS\mE^{m-1-j}.
\end{equation}
Clearly,   $\mD^{n-1-j}\mS\mD^{j}=\mD^{n-j-1}\mE \mD^{j}+\mD^{n}$ is the sum of monomials of degree $n$ and  $\mE^j\mS\mE^{m-1-j}=\mE^{m}+\mE^j\mD\mE^{m-1-j}$ is the sum  of monomials  of degree $m$.  We now multiply \eqref{ED-expansion1} by $\mX\mE^{m_{k+1}}$ from the left and use the induction assumption. Similarly, we multiply \eqref{ED-expansion2} by $\mD^{n_0}\mX$ from the right and use the induction assumption.  This establishes \eqref{ED-expansion} by induction.
\end{proof}
\begin{proposition}\label{P-positivity}
  If $\alpha\beta=q^N\gamma\delta$ then
  \begin{enumerate}
    \item $(-1)^L\varphi_0[(\mE+\mD)^L]>0$ for $L=0,\dots,N$
    \item $\varphi_1[(\mE+\mD)^L]>0$ for $L\geq N+1$.
  \end{enumerate}

\end{proposition}
\begin{remark}\label{Rem-signs}
 An inspection of our argument   shows that in the non-singular case with $\alpha\beta \ne q^n \gamma\delta$ for all $n$, we have $\varphi_0[(\mE+\mD)^L]\ne0$ for all $L$. More precisely, define $M=\min\{n\geq 0: \alpha\beta>q^n \gamma\delta\}$, with  $M=0$ when $\alpha\beta>\gamma\delta$.
Then %
    \begin{enumerate}
    \item $(-1)^L\varphi_0[(\mE+\mD)^L]>0$ for $0\leq L\leq M$
    \item $(-1)^M\varphi_0[(\mE+\mD)^L]>0$ for $L\geq M+1$.
  \end{enumerate}
In particular, the   current $J=\varphi_0[(\mE+\mD)^{L-1}]/\varphi_0[(\mE+\mD)^{L}]$ undergoes reversal as the system size increases: $J<0$ for $1\leq L\leq M$ and $J>0$ for $L\geq M+1$.
\end{remark}
\begin{proof}  Both proofs are similar and consist of showing that  for $\varphi=\varphi_0$ and for $\varphi=\varphi_1$
 the value $\varphi[\mX]$ on a monomial  $\mX=\mE^{m_1}\mD^{n_1}\dots \mE^{m_k}\mD^{n_k}$ is  real, and that for all monomials $\mX$ of the same degree $L=m+n$
 with   $m=m_1+\dots +m_k$, $n=n_1+\dots+n_k$,  the sign of $\varphi[\mX]$ is the same.
We begin  with the recursive proof for functional $\varphi=\varphi_0$ where the signs alternate with $L$. Then we will indicate how to modify the proof for $\varphi=\varphi_1$ where the signs are all positive.

For $L=0$  we have $(-1)^L\varphi[\mX]=1>0$  by the initialization of $\varphi_0$.
Suppose that  $(-1)^L\varphi[\mX]>0$ holds for all monomials $\mX=\mE^{m_1}\mD^{n_1}\dots \mE^{m_k}\mD^{n_k}$  with $m=m_1+\dots+m_k=m$, $n=n_1+\dots+n_k=n$ of degree $L=m+n <N$. %

A monomial $\mY$   of degree $L+1$ arises from a monomial $\mX$ of degree $L$ in one  of the following ways:
 $\mY=\mE\mX$,  $\mY=\mX\mD$,  $\mY=\mD\mX$, or   $\mY= \mX\mE$. Our goal is to show that in each of these cases
 $\varphi[\mY]$ is a real number of the opposite sign than $\varphi[\mX]$.

Cases   $\mY=\mE\mX$  and $\mY=\mX\mD$ are handled together, and are needed for the other two cases.
From \eqref{W+} and \eqref{V+} applied with $\mA=\mX$ %
we get
  $$
  \alpha\varphi[\mE\mX]-\gamma\varphi[\mD\mX]=\varphi[\mX] \mbox{ and }
  -\delta\varphi[\mX\mE]+\beta\varphi[\mX\mD]=\varphi[\mX].
  $$
Applying \eqref{ED-expansion} to $\mD\mX$ and  to $\mX\mE$  we get
\begin{eqnarray*}
    \alpha\varphi[\mE\mX]-q^{m}\gamma\varphi[ \mX\mD]=d_1 \\
  -q^{n}\delta\varphi[\mE\mX ]+\beta\varphi[\mX\mD]=d_2,
\end{eqnarray*}
where  by inductive assumption $d_1 =\varphi[\mX]+\gamma\sum_{j}c_j\varphi[\mZ_j]$ is the sum of non-zero real numbers  of the same sign $(-1)^L$, and similarly $ d_2 $ is real and has the  sign $(-1)^L$.
 The solution of this system is
\begin{equation}
  \label{sign-check}
  \varphi[\mE\mX]=\frac{\left|\begin{matrix}
  d_1 &-q^m \gamma \\
  d_2 & \beta
\end{matrix}\right|}{\left|\begin{matrix}
\alpha & -q^m\gamma \\
-q^n \delta & \beta
\end{matrix}\right|} \mbox{ and }  \varphi[ \mX\mD]=\frac{\left|\begin{matrix}
\alpha & d_1\gamma \\
-q^n \delta & d_2
\end{matrix}\right|}{\left|\begin{matrix}
\alpha & -q^m\gamma \\
-q^n \delta & \beta
\end{matrix}\right|}.
\end{equation}
Since the numerators have  sign $(-1)^L$ and the denominator $\alpha\beta-q^{L}\gamma\delta=\gamma\delta(q^N-q^L)<0$,
this establishes the conclusion for all monomials $\mY=\mE^{m_1+1}\mD^{n_1}\dots \mE^{m_k}\mD^{n_k}$ and
 $\mY=\mE^{m_1}\mD^{n_1}\dots \mE^{m_k}\mD^{n_k+1}$ of degree $m+n+1=L+1$.

To handle the case  $\mY=\mD\mX$, we use already established  information about the sign of monomial $\varphi[\mE\mX]$. Using \eqref{W+}, we see that the sign of
$\gamma\varphi[\mD\mX]=\alpha \varphi[\mE\mX]-\varphi[\mX]$  is $(-1)^{L+1}$,
and similarly \eqref{V+}  determines the sign of $\delta\varphi[\mX\mE]=\beta\varphi[\mX\mD]-\varphi[\mX]$ as $(-1)^{L+1}$.

The proof for $\varphi=\varphi_1$ is similar, starting with formula \eqref{phi-prod} which establishes positivity for $L=N+1$.   We then use \eqref{sign-check} to prove  that $\varphi_1[\mE\mX]>0$  and $\varphi_1[\mX\mD]>0$,
noting that in the case of $\varphi_1$ we have $d_1,d_2>0$ and that the denominator $\alpha\beta-q^{L}\gamma\delta=\gamma\delta(q^N-q^L)>0$ as $L\geq N+1$. Finally,  applying $\varphi_1$ to \eqref{ED-expansion} we see that  $\varphi_1[\mD\mX]>0$  and $\varphi_1[\mX\mE]>0$.
\end{proof}

\begin{proof}[Conclusion of proof of Theorem \ref{T-0}]
Functional $\varphi_0$ satisfies invariance conditions \eqref{V+} and \eqref{W+}, and   $\varphi_0\left[(\mE+\mD)^L\right]\ne 0$  for   $L\leq N$ by Proposition \ref{P-positivity}.
Therefore, by Theorem \ref{T2} we get \eqref{phi2P} for $L\leq N$. In the non-singular case, by Remark \ref{Rem:nonsingular} functional $\varphi_0$ is defined on  $\calA$ and by  Remark \ref{Rem-signs} we have  $\varphi_0[(\mE+\mD)^L]\ne 0$ for all $L$, so Theorem \ref{T2} applies.

Functional $\varphi_1$ satisfies invariance conditions \eqref{V+} and \eqref{W+} by Lemma \ref{L-phi12prod} and construction. Proposition \ref{P-positivity} states that $\varphi_1\left[(\mE+\mD)^L\right]>0$  for  $L\geq N+1$. Therefore, by Theorem \ref{T2} we get \eqref{phi2P}  for all   $L\geq N+2$.
Proposition \ref{P1} gives the stationary distribution for $L=N+1$, and Lemma \ref{L-phi12prod} shows that this case also arises from \eqref{phi2P}.

\end{proof}

\section{Proof of Theorem \ref{T3}}\label{sect:ProofT3}
Denote $\varphi_{k,n}=\varphi[\ee^k\dd^n]$, where $\varphi$ is either $\varphi_0$ or $\varphi_1$. (The latter is needed only for the second part of   Theorem \ref{T-Hermit2}.)
We first rewrite \eqref{SolA0} and \eqref{SolB0} using  Askey-Wilson parameters \eqref{AW-parameters}. After a calculation we get
\begin{multline}
  \label{SolA1}
\varphi_{m+1,n}=\frac{1}{1-abcd q^{m+n}}\Big(
\theta\left(c+d- cd (a+b)q^m\right)\varphi_{m,n}
\NL -   cd [m]_q\varphi_{m-1,n}+abcd q^m [n]_q\varphi_{m,n-1}
\Big),
\end{multline}
\begin{multline}
  \label{SolB1}
\varphi_{m,n+1}=\frac{1}{1-abcd q^{m+n}}\Big(\theta\left(a+b - ab(c+d)q^n\right)\varphi_{m,n}
\NL -ab [n]_q \varphi_{m,n-1}+abcd q^n [m]_q\varphi_{m-1,n}
\Big).
\end{multline}
\arxiv{In fact, it might be simpler to  use \eqref{AW-parameters} to rewrite \eqref{W+++} and \eqref{V+++} and then solve the
system of equations.

Notice that with \eqref{AW-parameters} equations \eqref{W++} and \eqref{V++} become
$$
\varphi[\ee^{m+1}\dd^n]+cd\varphi[\dd\ee^m\dd^n]=\theta (c+d)\varphi[\ee^m\dd^n],
$$
$$
ab \varphi[\ee^{m}\dd^n\ee]+ \varphi[\ee^m\dd^{n+1}]=\theta (a+b)\varphi[\ee^m\dd^n].
$$
}
Our proof relies heavily on monic continuous $q$-Hermite polynomials  defined by the three step recurrence
\begin{equation}\label{EQ-H}
x H_n(x)=H_{n+1}(x)+[n]_qH_{n-1}(x)
\end{equation}
with initial values $H_0(x)=1$ and $H_{-1}(x)=0$. These polynomials are convenient  because when   evaluated at $\ee+\dd$  they have
explicit expansion in the basis of monomials in normal order.

Somewhat more generally, for   $t\in\CC$ we consider polynomials $H_n(x;t)$ defined by the three step recurrence
\begin{equation}\label{q-H(t) recursion}
x H_n(x;t)=H_{n+1}(x;t)+t[n]_qH_{n-1}(x;t)
\end{equation}
with initial values $H_0(x;t)=1$ and $H_{-1}(x;t)=0$.
For $t\ne0$ these two families of polynomials
  are   related by a simple formula $H_n(x;t^2)=t^{n }H_n(x/t)$.

The following version of \cite[Corollary 2.8]{bozejko97qGaussian} follows  from \eqref{q-H(t) recursion}.
\begin{lemma}\label{L-BKS}
$$H_n(t\ee+\dd\;;t)=\sum_{k=0}^n \Binq{n}{k} t^k\ee^k\dd^{n-k}.$$
\end{lemma}

\begin{proof} Since $H_0(t\ee+\dd;t)=\mI$ and $H_1(t\ee+\dd;t)=t\ee\dd^0+\ee^0\dd$,
we only need to verify that the right hand side of the formula satisfies recursion \eqref{q-H(t) recursion}. That is, we have to show that
\begin{equation*}(t\ee+\dd)\sum_{k=0}^n\Binq{n}{k}  t^k\ee^k\dd^{n-k}-t [n]_q\sum_{k=0}^{n-1} \Binq{n-1}{k} t^k\ee^k\dd^{n-k}
\NL = \sum_{k=0}^{n+1} \Binq{n+1}{k}  t^k\ee^k\dd^{n+1-k} .
\end{equation*}
Using \eqref{de^m}, the left hand side is
\begin{multline*}
\sum_{k=0}^n\Binq{n}{k}  t^{k+1}\ee^{k+1}\dd^{n-k}+\sum_{k=0}^n\Binq{n}{k} t^k\dd \ee^k\dd^{n-k}-t[n]_q\sum_{k=0}^{n-1}
\Binq{n-1}{k} t^k\ee^k\dd^{n-1-k}
\\=
t^{n+1}\ee^{n+1}+\sum_{k=0}^{n-1}\Binq{n}{k}  t^{k+1}\ee^{k+1}\dd^{n-k}+\sum_{k=1}^nq^k\Binq{n}{k}  t^k
\ee^k\dd^{n+1-k}+\dd^{N+1}
\\ +
[n]_q\sum_{k=1}^n\Binq{n-1}{k-1}   t^k \ee^{k-1}\dd^{n-k} -[n]_q\sum_{k=1}^{n} \Binq{n-1}{k-1} t^{k}\ee^{k-1}\dd^{n-k}
\\=
t^{n+1}\ee^{n+1}+\sum_{k=1}^{n}\left(\Binq{n}{k-1}  +q^k\Binq{n}{k} \right)  t^k\ee^k\dd^{n+1-k}+\dd^{N+1}+0
\NL = \sum_{k=0}^{n+1} \Binq{n+1}{k} t^k \ee^k\dd^{n+1-k} ,
\end{multline*}
as
$$
\Binq{n}{k-1}+q^k\Binq{n}{k}%
=\Binq{n+1}{k}.$$
\end{proof}

We now introduce two   sequences of functions:
\begin{equation} \label{EQ-G}
G_n(t):=\varphi_0\left[ H_n(t\ee+\dd\;;t)\right],
\end{equation}
where $0\leq n<N+1$ (we include here non-singular case by allowing $N=\infty$), and
\begin{equation*} \label{EQ-F}
F_n(t):=\varphi_1\left[H_{n+N}(t\ee+\dd\;;t)\right], \; n\geq 1.
\end{equation*}
It turns out that these sequences satisfy similar recursions. %
\begin{theorem} \label{T-Hermit2}
For $0\leq n <N$ we have %
\begin{multline}\label{Eq:MArcin2}
  G_{n+1}(t)=\frac{\theta}{1-abcdq^n}\left((a+b)(1-t cd)G_n(qt)+(c+d)(t-q^n a b)G_n(t)\right) \\
  -\theta^2\frac{1-q^n}{1-a b c d q^n}\left(ab(1-t cd)G_{n-1}(qt)+t cd(t-a b q^n)G_{n-1}(t)\right)
\mathcal{}\end{multline}
with  $G_0(t)=1$ and  $G_{-1}(t)=0$.

For $n\geq 1$ we have
\begin{multline}\label{Eq:MArcin1}
F_{n+1}(t)=\frac{\theta}{1- q^{n }}\left((a+b)(1-t cd)F_n(qt)+(c+d)(t-q^{n+N} a b)F_n(t)\right) \\
-\theta^2\frac{1-q^{n+N}}{1-  q^{n }}\left(ab(1-t cd)F_{n-1}(qt)+t cd(t-a b q^{n+N})F_{n-1}(t)\right)
\end{multline}
with
$$F_1(t)= \frac{1}{\theta^{N+1}}\prod_{j=0}^{N}\frac{1-cd t q^j}{1-cd q^j} =\frac{(tcd;q)_{N+1}}{\theta^{N+1}(cd;q)_{N+1}}
$$ and   $F_{0}(t)=0$.
\end{theorem}
\begin{proof}%
Using the   identity $\Binq{n+1}{k}=q^k\Binq{n}{k}+\Binq{n}{k-1}$ we write
\begin{multline*}
  G_{n+1}(t)=\sum_{k=0}^{n+1}\Binq{n+1}{k}\varphi_{k,n+1-k} t^k
\NL  =\varphi_{0,n+1}+\sum_{k=1}^{n}\Binq{n+1}{k}\varphi_{k,n+1-k} t^k+\varphi_{n+1,0}t^{n+1}
  \\=\sum_{k=0}^{n}\Binq{n}{k} q^k\varphi_{k,n+1-k}+t \sum_{k=0}^n \Binq{n}{k}\varphi_{k+1,n-k} t^k = A+B \mbox{  (say)}.
\end{multline*}
Applying \eqref{SolB1} to expression $A$ we get

\begin{multline*}
 (1-abcd q^n) A
\NL = \sum_{k=0}^n \Binq{n}{k}q^k t^k \theta(a+b)\varphi_{k,n-k}
 -\sum_{k=0}^n \Binq{n}{k}q^kt^k \theta ab(c+d)q^{n-k}\varphi_{k,n-k}
 \\
 -ab \sum_{k=0}^n \Binq{n}{k}[n-k]_qq^k t^k \varphi_{k,n-1-k}+abcd \sum_{k=0}^n \Binq{n}{k}[k]_q q^k t^k q^{n-k}\varphi_{k-1,n-k}
 \\= \theta (a+b)G_n(qt)-\theta ab(c+d)q^n G_n(t)-ab [n]_q G_{n-1}(qt)+abcd [n]_q q^n t G_{n-1}(t).
\end{multline*}
Similarly applying    \eqref{SolA1} to expression $B$  we get
\begin{multline*}
  (1-abcd q^n)  B= t \theta (c+d)\sum_{k=0}^n\Binq{n}{k}t^k\varphi_{k,n-k}-cd (a+b)\theta t \sum_{k=0}^n \Binq{n}{k} q^k t^k \varphi_{k,n-k}\\
  -t cd \sum_{k=0}^n \Binq{n}{k} [k]_q t^k \varphi_{k-1,n-k}+t abcd \sum_{k=0}^n \Binq{n}{k}[n-k]_q t^k q^k \varphi_{k,n-1-k}
  \\= \theta t(c+d)G_n(t)-\theta t cd (a+b) G_n(qt)-t cd [n]_qG_{n-1}(t)+t abcd [n]_q G_{n-1}(qt).
\end{multline*}
Since $[n]_q=\theta^2(1-q^n)$,  we get \eqref{Eq:MArcin2}.

To determine the initial $F_1(t)$ we apply Lemma \ref{L-BKS} and formula \eqref{phi-1-ini} which in   parameters \eqref{AW-parameters} becomes  $$\varphi_{k,N+1-k}=\frac{(-cd)^kq^{\frac{k(k-1)}{2}}}{\theta^{N+1}\prod_{j=0}^{N}(1-cdq^j)}=\frac{1}{\theta^{N+1}(cd;q)_{N+1}}(-cd)^kq^{\frac{k(k-1)}{2}}.$$
We  have
\begin{multline*}
  F_1(t)=\varphi_1[H_{N+1}(t\ee+\dd;t)]= \sum_{k=0}^{N+1}\Binq{N+1}{k} t^{ k}\varphi_{k,N+1-k}
  \\=
 \frac{1}{\theta^{N+1}(cd;q)_{N+1}}\sum_{k=0}^{N+1}\Binq{N+1}{k} t^{ k} (-cd)^kq^{\frac{k(k-1)}{2}}
 =\frac{(tcd;q)_{N+1}}{\theta^{N+1}(cd;q)_{N+1}}
 ,
\end{multline*}
where we used Cauchy's $q$-binomial formula \eqref{pochtopot}. %
 The remaining steps of the proof are similar to the proof of recursion \eqref{Eq:MArcin2} and are omitted.

\end{proof}

\arxiv{
 For completeness, we include the omitted steps.
	\begin{multline*}
	F_{n+1}(t)=\sum_{k=0}^{N+n+1}{N+n+1\brack k}_q\varphi_{k,N+n+1-k}t^k\\
	=\varphi_{0,N+n+1}+\sum_{k=1}^{N+n}{N+n+1\brack k}_q\varphi_{k,N+n+1-k}t^k+\varphi_{N+n+1,0}t^{N+n+1}\\
	=\varphi_{0,N+n+1}+\sum_{k=1}^{N+n}\left({N+n\brack k}_q q^k+{N+n\brack k-1}_q\right)\varphi_{k,N+n+1-k}t^k+\varphi_{N+n+1,0}t^{N+n+1}\\
	=\sum_{k=0}^{N+n}{N+n\brack k}_q q^k \varphi_{k,N+n-k+1}t^k+t\sum_{k=0}^{N+n}{N+n\brack k}_q  \varphi_{k+1,N+n-k}t^k =A'+B' \mbox{ (say)}.
	\end{multline*}
}

\arxiv{
Applying \eqref{SolB1} to expression $A'$ we get
\begin{multline*}
	(1-q^{n}) A'= \sum_{k=0}^{n+N} \Binq{n+N}{k}q^k t^k \theta(a+b)\varphi_{k,n+N-k}
	-\sum_{k=0}^{n+N} \Binq{n+N}{k}q^kt^k \theta ab(c+d)q^{n+N-k}\varphi_{k,n+N-k}
	\\
	-ab \sum_{k=0}^{n+N} \Binq{n+N}{k}[n+N-k]_qq^k t^k \varphi_{k,n+N-1-k}+abcd \sum_{k=0}^{n+N} \Binq{n+N}{k}[k]_q q^k t^k q^{n+N-k}\varphi_{k-1,n+N-k}
	\\= \theta (a+b)F
	_n(qt)-\theta ab(c+d)q^{n+N} F_n(t)-ab [n+N]_q F_{n-1}(qt)+abcd [n+N]_q q^{n+N} t F_{n-1}(t)
	\end{multline*}
	Similarly applying    \eqref{SolA1} to expression $B'$  we get
	\begin{multline*}
	(1-q^{n})  B'= t \theta (c+d)\sum_{k=0}^{n+N}\Binq{n+N}{k} t^k\varphi_{k,n+N-k}-cd (a+b)\theta t \sum_{k=0}^{n+N} \Binq{n+N}{k} q^k t^k \varphi_{k,n+N-k}\\
	-t cd \sum_{k=0}^{n+N} \Binq{n+N}{k} [k]_q t^k \varphi_{k-1,n+N-k}+t abcd \sum_{k=0}^n \Binq{n+N}{k}[n+N-k]_q t^k q^k \varphi_{k,n+N-1-k}
	\\= \theta t(c+d)F_n(t)-\theta t cd (a+b) F_n(qt)-t cd [n+N]_qF_{n-1}(t)+t abcd [n+N]_q F_{n-1}(qt).
	\end{multline*}
	Since $[n+N]_q=\theta^2(1-q^{n+N})$,  we get \eqref{Eq:MArcin1}.
}

We now want to express the $q$-Hermite polynomials as linear combinations of the Askey-Wilson polynomials. We will start with
the following two explicit formulas for the connection coefficients,   relating $q$-Hermite polynomials  with Al-Salam-Chihara
polynomials in the first step, and then with Askey-Wilson polynomials in the second step. (This topic is well studied, see e.g.
\cite{foupouagnigni2013connection,tcheutia2014connection} and the references therein, so both formulas should be known; but we
were not able to locate them in the literature.)
\begin{proposition}\label{HerToAWCon}
For $a,b\in\CC$, the connection coefficients in the expansion
 \begin{equation}p_n(x;0,0,0,0|q)=\sum_{k=0}^{n}c_{n,k}p_k(x;a,b,0,0|q)\label{AW2H0}
  \end{equation}
  are
    \begin{equation}
  \label{c_nk}
  c_{n,k}=\sum_{\ell=k}^n \Binq{n}{\ell}\Binq{\ell}{k}a^{n-\ell} b^{\ell-k}.
  \end{equation}
If $a\ne 0$, the connection coefficients %
 in the expansion  $$p_n(x;0,0,0,0|q)=\sum_{\ell=0}^ne_{n,\ell}(a,b,c,d)p_\ell(x;a,b,c,d|q)$$ are
    $$e_{n,\ell}(a,b,c,d)=\sum_{k=\ell}^nc_{n,k}{k \brack \ell}_q\frac{q^{\ell(\ell-k)}(abq^\ell;q)_{k-\ell}}{a^{k-\ell}(abcdq^{\ell-1};q)_\ell}{_3\phi_2}\left(\begin{matrix}
	q^{\ell-k},\ acq^\ell,\ adq^\ell\\\ 0,\ abcdq^{2\ell}
\end{matrix}\middle|q;q\right).$$
\end{proposition}

\begin{proof}
Since \eqref{AW2H0} holds trivially when $a=b=0$, by symmetry of $p_k(x;a,b,0,0|q)$ in parameters $a,b$, we can assume $a\ne 0$.  From \eqref{conectiocoefficient} we see that
\begin{equation}
  \label{AW2H1} p_n(x;a,0,0,0|q)=\sum_{k=0}^n C_{n,k}p_k(x;a,b,0,0|q),
\end{equation}
where
\begin{multline*}
   C_{n,k}=\frac{q^{k(k-n)}}{a^{n-k}}{n \brack k}_q {_2\phi_1}\left(\begin{matrix}
q^{k-n},\ abq^k\\ 0
\end{matrix}\middle|q;q\right)=\frac{q^{k(k-n)}}{a^{n-k}}{n \brack k}_q(abq^k)^{n-k}
\\={n \brack k}_qb^{n-k}
\end{multline*}
(we used formula \eqref{twophione}.)
In particular \eqref{AW2H1} is valid also for $a=0$. Setting $a=0$ in \eqref{AW2H1}, using symmetry again, and renaming $b$ as $a$ we get
\begin{equation}
\label{AW2H2} p_n(x;0,0,0,0|q)=\sum_{k=0}^n {n \brack k}_qa^{n-k}p_k(x;a,0,0,0|q).
\end{equation}
Combining \eqref{AW2H2} with \eqref{AW2H1} proves that
\begin{equation*}
\label{AW2H3} p_n(x;0,0,0,0|q)=\sum_{k=0}^n c_{n,k} p_k(x;a,b,0,0|q),
\end{equation*} where $c_{n,k}$ is given by \eqref{c_nk}. This formula holds for all $a,b$.

Next we prove the second connection formula for $a\ne 0$.
From \eqref{conectiocoefficient} follows that the coefficient $C'_{k,\ell}$ in the expansion
\begin{equation*}
\label{AW2H4} p_k(x;a,b,0,0|q)=\sum_{\ell=0}^n C'_{k,\ell} p_\ell(x;a,b,c,d|q)
\end{equation*}
is equal to $${k \brack \ell}_q\frac{q^{\ell(\ell-k)}(abq^\ell;q)_{k-\ell}}{a^{k-\ell}(abcdq^{\ell-1};q)_\ell}{_3\phi_2}\left(\begin{matrix}
	q^{\ell-k},\ acq^\ell,\ adq^\ell\\\ 0,\ abcdq^{2\ell}
\end{matrix}\middle|q;q\right).$$
This ends the proof, since $e_{n,\ell}(a,b,c,d)=\sum_{k=\ell}^nc_{n,k}C'_{k,\ell}$.
\end{proof}

 Suppose that the degrees of polynomials $p_k$ are $k$ for $k=0,1,\dots,n$. (Recall that
this fails for large $n$ if $q^Nabcd=1$ for some $N=0,1,\dots$.) Denote by $\{a_{n,k}(a,b,c,d)\}$
the coefficients in the expansion
\begin{equation}
  \label{H2P}
  H_n(x)=\sum_{k=0}^n a_{n,k}(a,b,c,d)p_k\left(\frac{x}{2\theta};a,b,c,d\middle| q\right),
\end{equation}
where $H_n(x)=H_n(x;1)$ is given by  \eqref{EQ-H}.

We will need explicit formula for
the  coefficient $A_n(a,b,c,d):=a_{n,0}(a,b,c,d)$.  %
Since $a_{n,k}(a,b,c,d)$ are invariant under permutations of $a,b,c,d$, without loss of generality we assume $a\ne 0$.
This is enough for our purposes, as we have $a,c>0$  for the parameters arising from ASEP.
\begin{proposition}
\begin{equation*}
  \label{Marcin-SolA}
  A_n(a,b,c,d)=\theta^n\sum_{k=0}^nc_{n,k}\frac{(ab;q)_k}{a^k}{_3\phi_2}\left(\begin{matrix}
  q^{-k},\ ac,\ ad\\\ 0,\ abcd
  \end{matrix}\middle|q;q\right),
  \end{equation*}
  with $c_{n,k}$ given by \eqref{c_nk}.
\end{proposition}
\begin{proof}
By comparing the three step recursions, it is clear that $H_n(x)=\theta^n p_n(\frac{x}{2\theta};0,0,0,0|q)$. Hence, by Proposition \ref{HerToAWCon},  $A_n(a,b,c,d)=\theta^n e_{n,0}(a,b,c,d)$.
\end{proof}

It turns out that   $A_n(a,b,c,d)$ is related to the moment of the $n$-th $q$-Hermite polynomial  introduced in \eqref{EQ-G}.
\begin{proposition}\label{L:A=G}   %
For $0\leq n<N$, $a,c>0$ and $t\ne 0$ we have
$$t^n G_n(1/t^2)=A_n(at,bt,c/t,d/t).$$
\end{proposition}

For the proof, we need to rewrite both sides of this equation.

For the next lemma, we write $G_{n}(z)$ as $G_n(z;a,b,c,d)$ with explicitly written Askey-Wilson parameters.
In this notation, Proposition \ref{L:A=G} says $$t^n G_n(1/t^2;a,b,c,d)=A_n(at,bt,c/t,d/t),$$ which is the same as
$t^n G_n(1/t^2;a/t,b/t,ct,dt)=A_n(a,b,c,d)$.
\begin{lemma}Expression
\begin{equation}
  \label{B-sub} B_n(a,b,c,d):=
   (abcd;q)_n\frac{G_n(t^2; at, bt, c/t,d/t)}{\theta^n t^n}
\end{equation}  does not depend on $t$ and satisfies the
following recursion %
for $0\leq n <N$:
  \begin{multline}\label{B-rec}
     B_{n+1}( a,b,c,d)
 \NL = (a +b)(1-cd)q^{n/2}B_n(
      a/\sqrt{q},b/\sqrt{q},c\sqrt{q},d\sqrt{q})+(c+d)(1-q^n ab)  B_n(
      a,b,c,d)\\
     -(1-q^n)(1-abcd q^{n-1}) \Big(ab (1-cd)q^{(n-1)/2}B_{n-1}(
      a/\sqrt{q},b/\sqrt{q},c\sqrt{q},d\sqrt{q})
      \\+cd (1-ab q^n)B_{n-1}(a,b,c,d)\Big)
  \end{multline}
  with the initial value  $B_0(a,b,c,d)=1$, and  $B_{-1}(a,b,c,d)=0$.
\end{lemma}

\begin{proof}
  Denote by $\widetilde G_n(t^2;a,b,c,d)$ the right hand side of \eqref{B-sub}. Inserting this expression into \eqref{Eq:MArcin2} we get
  recursion
  \begin{multline}\label{tildeG}
     \widetilde G_{n+1}(t^2;a,b,c,d)=(a +b)(1-cd)\widetilde G_n(q t^2;a,b,c,d)
   \NL +(c+d)(1-q^n ab) \widetilde G_n(
     t^2;a,b,c,d)\\
     -(1-q^n)(1-abcd q^{n-1}) \Big(ab (1-cd) \widetilde G_{n-1}(q t^2;a,b,c,d)
  \NL  +cd (1-ab q^n)\widetilde G_{n-1}(  t^2;a,b,c,d)\Big)
  \end{multline}
  with the coefficients that do not depend on $t$. Since the initial condition $\widetilde G_{-1}=0$ and $\widetilde G_{0}=1$ does
  not depend on $t$, therefore the solution of the recursion does not depend on $t$. We check this by induction, assuming that this assertion holds for $\widetilde G_0,\dots,\widetilde G_n$. Denoting $\tilde t=t\sqrt{q}$ we have
 \begin{multline*}
\widetilde G_n(q t^2;a,b,c,d)= (abcd;q)_n\frac{G_n(q t^2; at, bt, c/t,d/t)}{\theta^n t^n} = \\
=q^{n/2} (abcd;q)_n\frac{G_n(\tilde t ^2; \frac{a}{\sqrt{q}}\tilde t, \frac{b}{\sqrt{q}}\tilde t, \sqrt{q} c/\tilde t,\sqrt{q}d/\tilde
t)}{\theta^n \tilde t^n }
\NL =q^{n/2} B_n( a/\sqrt{q},b/\sqrt{q},c\sqrt{q},d\sqrt{q}).
\end{multline*}
Thus \eqref{tildeG} shows that $\widetilde G_{n+1}(t^2;a,b,c,d)$ does not depend on $t$, and recursion \eqref{B-rec} follows.
\end{proof}

Next we rewrite the %
right hand side of the equation in Proposition \ref{L:A=G}. Denote
\[\widetilde A_n(a,b,c,d)=
  (abcd;q)_n A_n(a,b,c,d)/\theta^n
  \NL = (abcd;q)_n \sum_{k=0}^nc_{n,k}\frac{(ab;q)_k}{a^k}{_3\phi_2}\left(\begin{matrix}
  q^{-k},\ ac,\ ad\\\ 0,\ abcd
  \end{matrix}\middle|q;q\right).
  \]
  We rewrite this as
$$
\widetilde A_n(a,b,c,d)=(abcd;q)_n\sum_{k=0}^n   (ab;q)_k   c_{n,k}\beta_k
$$ with
$$\beta_k(a,b,c,d)=\frac{1}{a^k}{_3\phi_2}\left(\begin{matrix}
  q^{-k},\ a d,\  a c\\0 ,\  a b c d
\end{matrix}\middle|q;q\right)=\frac{1}{a^k}\sum_{j=0}^k \frac{(q^{-k}, ad, ac;q)_j}{(q, abcd;q)_j}q^j.$$
In order to prove Proposition \ref{L:A=G} it is enough to show that $\widetilde A_n(a,b,c,d)=B_n(
      a,b,c,d)$.  Since both expressions are $1$ when $n=0$, we only need to verify that $\widetilde A_n(a,b,c,d)$ satisfies recursion \eqref{B-rec}.
To accomplish this goal, we need auxiliary recursions for the coefficients $c_{n,k}$ and $\beta_k$.
\begin{lemma}\label{lemmarek1}
With   the usual convention that $c_{n,k}=0$ if $k>n$ or $k<0$, for all $n\geq 0$ and all $k$, we have
\begin{equation}
  \label{lemarek1A}
  c_{n+1,k}=c_{n,k-1}+q^k(a+b)c_{n,k}-q^k(1-q^n)ab\cdot c_{n-1,k}.
\end{equation}
Furthermore, for  $n\geq 1$ and $ 0\leq k\leq n$ we have
\begin{equation}
    \label{lemarek1B}
  (1-q^{k+1})c_{n,k+1}=(1-q^n)c_{n-1,k}.
\end{equation}
\end{lemma}
\begin{proof}
Let  $h_{n}(x)=p_n(x;0,0,0,0|q)$ and $Q_n(x)=p_n(x;a,b,0,0|q)$.  Then \eqref{AW2H0} is
$$h_n(x)=\sum_{k=0}^n c_{n,k}Q_k(x), \; n\geq 0.$$
Comparing the three step recursions
$$2xh_n(x)=h_{n+1}(x)+(1-q^n)h_{n-1}(x)$$
and
\begin{equation}
  \label{rec-A-Chi}
  2xQ_n(x)=Q_{n+1}(x)+q^n(a+b)Q_n(x)+(1-q^n)(1-q^{n-1}ab)Q_{n-1}(x),
\end{equation}
see, e.g., \cite[(3.8.4)]{koekoek1998askey},   we get
\begin{equation}
  \label{c***}
  c_{n+1,k}=c_{n,k-1}+q^k(a+b)c_{n,k}
  \NL +(1-q^{k+1})(1-q^kab)c_{n,k+1}
 -(1-q^n)c_{n-1,k}.
\end{equation}
	Indeed, expanding both sides of $2xh_n(x)=h_{n+1}(x)+(1-q^n)h_{n-1}(x)$ and applying \eqref{rec-A-Chi} to the expansion on left hand side we get
\begin{multline*}
\sum_{k=0}^n c_{n,k}\left(Q_{k+1}(x)+q^k(a+b)Q_k(x)+(1-q^k)(1-q^{k-1}ab)Q_{k-1}(x)
\right)
\\= \sum_{k=0}^{n+1}c_{n+1,k}Q_k(x)+(1-q^n)\sum_{k=0}^{n-1}c_{n-1,k}Q_k(x).
\end{multline*}
	The formula follows by comparing the coefficients at $Q_k(x)$.

Since $c_{n,k}=c_{n,k}(a,b)$ is a homogeneous polynomial of degree $n-k$  in variables $a$ and $b$, we can separate the
components of
recursion \eqref{c***} into the  pair of recursions. The terms of degree $n-k-1$ give  \eqref{lemarek1B}. The terms of degree
$n+1-k$ give $c_{n+1,k}=c_{n,k-1}+q^k(a+b)c_{n,k}-(1-q^{k+1})q^k a b \cdot c_{n,k+1}$, which gives \eqref{lemarek1A} after
using \eqref{lemarek1B}.
\end{proof}
\begin{corollary}\label{corrolaryocnk}
$$(ab;q)_kc_{n+1,k}-(1-q^nab)(ab;q)_{k-1}c_{n,k-1}=(a+b)\left(\frac{ab}{q};q\right)_kq^kc_{n,k}
\NL -(1-q^n)ab \left(\frac{ab}{q};q\right)_kq^kc_{n-1,k}.
$$
\end{corollary}
\begin{proof}
	It is enough to prove that
$$(1-abq^{k-1})c_{n+1,k}=(a+b)\left(1-\frac{ab}{q}\right)q^k c_{n,k}+(1-q^nab)c_{n,k-1}
\NL-(1-q^n)ab\left(1-\frac{ab}{q}\right)q^kc_{n-1,k}.
$$
	Since $c_{n,k}$ is a homogeneous polynomial of degree $n-k$  in variables $a$ and $b$ this is equivalent to a pair of identities
	\begin{equation}
	c_{n+1,k}=q^k(a+b)c_{n,k}+c_{n,k-1}-q^k(1-q^n)ab\cdot c_{n-1,k},\label{corrolaryocnkeq1}
	\end{equation}
which is 	\eqref{lemarek1A}, and	
\begin{equation}-abq^{k-1}c_{n+1,k}=-ab(a+b)q^{k-1}c_{n,k}-q^nab\cdot c_{n,k-1}
\NL+(1-q^n)a^2b^2q^{k-1}c_{n-1,k}.
\label{corrolaryocnkeq2}
	\end{equation}
 To prove  \eqref{corrolaryocnkeq2} it is enough to verify that
	\begin{equation*}
	q^kc_{n+1,k}=q^k(a+b)c_{n,k}+q^{n+1}c_{n,k-1}-q^k(1-q^n)ab\cdot c_{n-1,k}.
	\end{equation*}
	To do this, we subtract this expression from \eqref{corrolaryocnkeq1} and use %
 \eqref{lemarek1B}.
\arxiv{
We get $(1-q^k)c_{n+1,k}= (1-q^{n+1})c_{n,k-1} $.
}
\end{proof}
We also need the following recursion which was discovered by Mathematica  package  {\tt qZeil} \cite{paule1997mathematica},
but for which we have a standard  proof.
\begin{lemma}\label{lammarek2}  For $0\leq n <N$, $a\ne 0$ and $b,c,d\in\CC$  we have
\begin{equation}
  \label{qZeil} (1-a b c d q^{n}) \beta_{n+1}(a,b,c,d)=  \left(c  +d -  c d (a+b) q^{n}\right) \beta_{n}(a,b,c,d)
  \NL -	c d \left(1-q^{n}\right)  \beta_{n-1}(a,b,c,d).
\end{equation}
	The initial condition for this recursion is  $\beta_{0}=1, \beta_{-1}=0$.
\end{lemma}
\begin{proof}
For $a\ne 0$, consider the Al-Salam--Chihara polynomials
\begin{equation}
  \label{Q-AlSalam}
  \widetilde Q_n(x;a,b)=\frac{a^n}{(ab;q)_n}p_n(x;a,b,0,0|q)={_3}\phi_2\left(\begin{matrix}
q^{-n},a e^{i\psi},a e^{-i\psi} \\
0,\ ab
\end{matrix}\middle|q;q\right),
\end{equation} where $x=\cos\psi$. %
The three step recursion for polynomials $\widetilde Q_n(x)$ is
\begin{multline}\label{ThreeStepKoek}
2x\widetilde Q_n(x;a,b)=a^{-1}(1-abq^n)\widetilde Q_{n+1}(x;a,b)
+(a+b)q^n\widetilde Q_n(x;a,b)
\NL+a(1-q^n)\widetilde Q_{n-1}(x;a,b)\end{multline}
with $\widetilde Q_0(x;a,b)=1$ and $\widetilde Q_{-1}(x;a,b)=0$. (This is a version of \eqref{rec-A-Chi} under different normalization.)
For $c,d>0$ let $x_*=\frac12\left(\sqrt{\tfrac{c}{d}}+\sqrt{\tfrac{d}{c}}\right)$.
It is easy to see that
$$\widetilde Q_n\left(x_*;a\sqrt{cd},b\sqrt{cd}\right)={_3}\phi_2\left(\begin{matrix}
q^{-n},\ ac,\ ad \\
0,\ abcd
\end{matrix}\middle|q;q\right)=a^n\beta_n(a,b,c,d).$$
Indeed, to  extend  polynomial $\widetilde Q_n(x)$  from $x=\cos \psi\in[-1,1]$ to $x>1$ we
  replace  $e^{\pm i\psi}$ in \eqref{Q-AlSalam} by $x\pm \sqrt{x^2-1}$. These  expressions evaluate to $\sqrt{c/d}$
  and $\sqrt{d/c}$ at $x=x_*$. %

Recursion  \eqref{ThreeStepKoek} implies that %
$$
   \left(\sqrt{\tfrac{c}{d}}+\sqrt{\tfrac{d}{c}}\right)a^n\beta_n=
   \frac{1}{a\sqrt{cd}}(1-abcdq^n)a^{n+1}\beta_{n+1}
+\left(a\sqrt{cd}+b\sqrt{cd}\right)q^na^n\beta_n
\NL+a\sqrt{cd}(1-q^n)a^{n-1}\beta_{n-1}.
$$
This implies \eqref{qZeil} for $a\ne 0$ and $c,d>0$. We now use the fact that $\beta_{n}(a,b,c,d)$ is a rational function of $a,b,c,d$,
  with the denominator that has factors $a^k$ and  $1-abcd q^k$, $0\leq k\leq n<N$. Thus recursion \eqref{qZeil} extends to all $a,b,c,d$
 within the domain of $\beta_{n}(a,b,c,d)$.
\end{proof}
 \begin{proof}[Proof of Proposition \ref{L:A=G}]
We will show that $$\widetilde A_n(a,b,c,d):=(abcd;q)_n\sum_{k=0}^n   (ab;q)_k   c_{n,k}\beta_k$$ satisfies recursion \eqref{B-rec}.
We first note that
$$c_{n,k}(a/\sqrt{q},b/\sqrt{q})=q^{(k-n)/2}c_{n,k}(a,b)$$ and $$ \beta_k(a/\sqrt{q},b/\sqrt{q},c\sqrt{q},d\sqrt{q})=q^{k/2}\beta_k(a,b,c,d).$$
We therefore want to show that
	\begin{eqnarray*}
	\frac{\widetilde A_{n+1}(a,b,c,d)}{(abcd;q)_n}&=&(a+b)(1-cd)\sum_{k=0}^{n}\left(\frac{ab}{q};q\right)_kq^kc_{n,k}\beta_k\\
	&&+(c+d)(1-q^nab)\sum_{k=0}^{n}\left(ab;q\right)_kc_{n,k}\beta_k\\
	&&-(1-q^n)ab(1-cd)\sum_{k=0}^{n-1}\left(\frac{ab}{q};q\right)_kq^kc_{n-1,k}\beta_k\\
	&&-(1-q^n)cd(1-abq^{n})\sum_{k=0}^{n}\left(ab;q\right)_kc_{n-1,k}\beta_k.
\end{eqnarray*}
We will be working with the right hand side of this equation. The sum of the first and the third term is equal to
$$(1-cd)\sum_{k=0}^n\left[(a+b)\left(\frac{ab}{q};q\right)_kq^kc_{n,k}-(1-q^n)ab\left(\frac{ab}{q};q\right)_kq^k c_{n-1,k}\right]\beta_k.$$
By Corollary \ref{corrolaryocnk} this is equal
$$(1-cd)\sum_{k=0}^{n}\left(ab;q\right)_kc_{n+1,k}\beta_k-(1-cd)(1-abq^n)\sum_{k=0}^{n}\left(ab;q\right)_{k-1}c_{n,k-1} \beta_k=$$
$$=(1-abcdq^n)\sum_{k=0}^{n}\left(ab;q\right)_kc_{n+1,k}\beta_k-cd(1-abq^n)\sum_{k=0}^{n}\left(ab;q\right)_kc_{n+1,k}\beta_k$$
$$-(1-cd)(1-abq^n)\sum_{k=0}^{n}\left(ab;q\right)_{k-1}c_{n,k-1}\beta_k,$$
since $(1-cd)=(1-abcdq^n)-cd(1-abq^n)$.\\

It follows that what we want to show is
$$	\frac{\widetilde A_{n+1}(a,b,c,d)}{(abcd;q)_n}=(1-abcdq^n)\sum_{k=0}^{n}\left(ab;q\right)_kc_{n+1,k}\beta_k+(1-abq^n) S,$$
where %
\begin{multline*}
S=S_1-S_2-S_3-S_4
\NL=\overbrace{(c+d)\sum_{k=0}^{n}\left(ab;q\right)_kc_{n,k}\beta_k}^{S_1}-\overbrace{(1-q^n)cd\sum_{k=0}^{n}\left(ab;q\right)_kc_{n-1,k}\beta_k}^{S_2}
\\-\overbrace{cd\sum_{k=0}^{n}\left(ab;q\right)_kc_{n+1,k}\beta_k}^{S_3}-\overbrace{(1-cd)\sum_{k=0}^{n}\left(ab;q\right)_{k-1}c_{n,k-1}\beta_k}^{S_4}.
\end{multline*}
We will finish the proof by showing that $S$ is equal to $(1-abcdq^n)\left(ab;q\right)_{n}c_{n+1,n+1}\beta_{n+1}$.

By Lemma \ref{lemmarek1}
\begin{multline*}S_3=cd\sum_{k=0}^{n}\left(ab;q\right)_k c_{n+1,k}\beta_k=S'_3+S''_3-S'''_3
\NL=
\overbrace{cd(a+b)\sum_{k=0}^{n}\left(ab;q\right)_kq^kc_{n,k}\beta_k}^{S'_3}+
\overbrace{cd\sum_{k=0}^{n}\left(ab;q\right)_kc_{n,k-1}\beta_k}^{S''_3}
\\-\overbrace{cd\sum_{k=0}^{n}\left(ab;q\right)_kq^k(1-q^n)ab\cdot c_{n-1,k}\beta_k}^{S'''_3}.
\end{multline*}
Since $cd\left(ab;q\right)_k=cd\left(ab;q\right)_{k-1}-abcdq^{k-1}\left(ab;q\right)_{k-1}=-(1-cd)\left(ab;q\right)_{k-1}+(1-abcdq^{k-1})\left(ab;q\right)_{k-1}$ we see that %
\begin{multline*}
 S''_3=-\overbrace{(1-cd)\sum_{k=0}^{n}\left(ab;q\right)_{k-1}c_{n,k-1}\beta_k}^{S_4}\NL
 +\overbrace{\sum_{k=0}^{n}(1-abcdq^{k-1})\left(ab;q\right)_{k-1}c_{n,k-1}\beta_k}^{I}
 =-S_4+I.
\end{multline*}
Writing $abq^k=-(1-abq^k)+1$ we can rewrite $S'''_3$ as
\begin{multline*}
S'''_3=-cd\sum_{k=0}^{n}\left(ab;q\right)_{k+1} \underbrace{(1-q^n)c_{n-1,k}}_{\textrm{Lemma}\  \ref{lemmarek1}}\beta_k +(1-q^n)cd\sum_{k=0}^{n}\left(ab;q\right)_{k} \beta_k\\
=-\overbrace{cd\sum_{k=0}^{n}\left(ab;q\right)_{k+1} (1-q^{k+1})c_{n,k+1}\beta_k}^{J} \NL
+\overbrace{(1-q^n)cd\sum_{k=0}^{n}\left(ab;q\right)_{k}c_{n-1,k} \beta_k}^{S_2}=-J+S_2.
\end{multline*}
Combining all the expressions together we obtain
$$S=\left(S_1-S'_3-J\right)-I.$$
The first expression is equal
\begin{multline*}
   S_1-S'_3-J\NL=\overbrace{\sum_{k=0}^{n}\left(ab;q\right)_kc_{n,k}\left[c+d-cd(a+b)q^k\right]\beta_k}^{S_1-S'_3}
   -\overbrace{cd\sum_{k=0}^{n}\left(ab;q\right)_{k}(1-q^{k})c_{n,k}\beta_{k-1}}^{J}
\\=\sum_{k=0}^{n}\left(ab;q\right)_kc_{n,k}\underbrace{\left\{\left[c+d-cd(a+b)q^k\right]\beta_k-cd(1-q^k)\beta_{k-1}\right\}}_{
\textrm{RHS of } \eqref{qZeil}} \NL
=\sum_{k=0}^{n}\left(ab;q\right)_kc_{n,k}(1-abcdq^k)\beta_{k+1}.\end{multline*}
Hence
\begin{multline*}S=\overbrace{\sum_{k=0}^{n}\left(ab;q\right)_kc_{n,k}(1-abcdq^k)\beta_{k+1}}^{S_1-S'_3-J}-
\overbrace{\sum_{k=0}^{n}(1-abcdq^{k-1})\left(ab;q\right)_{k-1}c_{n,k-1}\beta_k}^{I}
\NL =(ab;q)_nc_{n,n}(1-abcdq^n)\beta_{n+1}.\end{multline*}
This ends the proof, as $c_{n+1,n+1}=c_{n,n}=1$.
 \end{proof}

\begin{proof}[Proof of Theorem \ref{T3}] The proof does not use explicitly singularity condition  $q^Nabcd=1$, except for the constraints that it implies on the domain of $\varphi_0$ and on the
degrees of the polynomials $\{p_k: k=1,\dots,N\}$.

For $n=1$ this is a calculation, which  is also covered  by the induction step.
 Suppose that $p_k$ is of degree $k$ and
 $$\varphi_0\left[p_k\left( \xx_t;at,bt,c/t,d/t\middle|q\right)\right]=0 \mbox{  for  $k=1,\dots, n$.}$$
  Suppose that  polynomial $p_{n+1}$ is of  degree
$n+1$. Then, recalling  \eqref{ed2x}, we have
\begin{multline*}
 H_{n+1}(\ee/t^2+\dd\;;1/t^2)= H_{n+1}(2\theta \xx_t/t;1/t^2)= \frac{1}{t^{n+1}}H_{n+1}(2\theta \xx_t)
 \\= \frac{1}{t^{n+1}}
\sum_{k=0}^{n+1} a_{n+1,k}(at,bt,c/t,d/t)p_k\left(\xx_t;at,bt,c/t,d/t\middle|q\right)
\end{multline*}
by \eqref{H2P}. Since $p_0=1$, by inductive assumption we have
\begin{multline*}
   \varphi_0\left[H_{n+1}(\ee/t^2+\dd\;;1/t^2) \right] =\frac{1}{t^{n+1}}a_{n+1,0}(a t,bt,c/t,d/t)
   \NL+\frac{1}{t^{n+1}}a_{n+1,n+1}\varphi_0\left[ p_{n+1}\left(\xx_t;at,bt,c/t,d/t\middle|q\right)\right].
\end{multline*}
This shows that $\varphi_0[p_{n+1}\left(\xx_t;at,bt,c/t,d/t\middle|q\right)]=0$, provided that $a_{n+1,n+1}\ne 0$, which holds true due to
the assumption  on the degree of $p_{n+1}$, and provided that
\begin{equation*}
  a_{n+1,0}(a t,bt,c/t,d/t)=t^{n+1}G_{n+1}(1/t^2),
\end{equation*}
which holds true by Proposition \ref{L:A=G}.

Since the degree of polynomial $p_n$ is $n$ for $n\leq  \lfloor(N+1)/2\rfloor$, this establishes the conclusion such $n$. For $n>\lfloor(N+1)/2\rfloor$, polynomial $p_n$ is a constant multiple of polynomial $p_{N+1-n}$, so the conclusion also holds.
\end{proof}

\section{Conclusions}
In this paper we  construct   a functional $\varphi_0$, or a pair of functionals $\varphi_0,\varphi_1$,  on an abstract algebra   that give stationary probabilities for an ASEP of length $L$ with arbitrary parameters. Formula \eqref{MatrixSolution+} for the probabilities extends the celebrated matrix product ansatz \citep{derrida1993exact} to the singular case with $\alpha\beta=q^N\gamma\delta$.  Our  approach avoids an associativity pitfall that may arise in matrix product models. In  Appendix \ref{Sec:MatrixModel} we exhibit  an example of such a  matrix model that  satisfies the usual conditions \eqref{q-comm-Derrida} \eqref{W} \eqref{V},  yet it cannot be used to compute stationary probabilities.

While verifying that our functionals give non-zero answers for un-normalized probabilities, we noted an interesting phenomenon of current reversal   as the system size $L$ increases   when $\alpha\beta<\gamma\delta$ and $0<q<1$ .

In the non-singular case, we prove that functional $\varphi_0$ may serve as an  orthogonality functional for the Askey-Wilson polynomials with fairly general parameters. Part of this connection persists in the singular case $\alpha\beta=q^N\gamma\delta$ when the degrees of the first $N$ Askey-Wilson polynomials do not exceed $(N+1)/2$. In Appendix \ref{Sect:TASEP} we give explicit formulas for the (formal) Cauchy-Stieltjes transforms of both functionals when   $q=0$.

\subsection*{Acknowledgements}
The authors thank Peter Paule   for sharing mathematica software packages {\tt qZeil} and {\tt qMultiSum}
 developed in Research Institute for Symbolic Computation at the University of Linz. They thank   Daniel Tcheutia for helpful
 comments on the early draft of the paper and Alexei Zhedanov for references.
Finally the authors thank the referees for the thorough and informative reviews that helped us to improve the paper.

Marcin \'Swieca's research was partially supported by grant 2016/21/B/ST1/00005 of National Science Centre, Poland.

\appendix
\section{Auxiliary identities}
Here we collect $q$-hypergeometric formulas used in this paper. Cauchy's $q$-binomial formula is
\begin{equation}\label{pochtopot}
 (x;q)_n=\sum_{k=0}^n{n \brack k}_q(-1)^kq^{\frac{k(k-1)}{2}}x^k.
\end{equation}
Heine's summation formula   \cite[(1.5.3)]{gasper2004basic} reads
\begin{equation}
{_2\phi_1}\left(\begin{matrix}
q^{-n},\ b\\ c
\end{matrix}\middle|q;q\right)=\frac{(c/b;q)_n}{(c;q)_n}b^n.\label{twophione}
\end{equation}
We also need   the connection coefficients of the Askey-Wilson polynomials.
\begin{theorem}[\cite{Askey-Wilson-85}]
If $a_4\neq 0$ then %
$$p_n(x;b_1,b_2,b_3,a_4|q)=\sum_{k=0}^nc_{n,k}p_k(x;a_1,a_2,a_3,a_4|q),$$
where
\begin{multline}\label{conectiocoefficient}
c_{n,k}=(b_1b_2b_3a_4;q)_k\frac{q^{k(k-n)}(q;q)_n(b_1a_4q^k,b_2a_4q^k,b_3a_4q^k;q)_{n-k}}{a_4^{n-k}(q;q)_{n-k}(q,a_1a_2a_3a_4q^{k-1};q)_k}
\\ \times{_5\phi_4}\left(\begin{matrix}
q^{k-n},\ b_1b_2b_3a_4q^{n+k-1},\  a_1a_4q^k,\ a_2a_4q^k,\ a_3a_4q^k\\b_1a_4q^k, \ \  b_2a_4q^k,\ \ b_3a_4q^k,\ \  a_1a_2a_3a_4q^{2k}
\end{matrix}\middle|q;q\right).
\end{multline}
\end{theorem}

\section{Totally asymmetric case}\label{Sect:TASEP}
Our recursions simplify  when  $q=0$, i.e.,  the case of Totaly Asymmetric Exclusion Process.
Then the conclusion of Theorem \ref{T3} can be derived more directly, and there is also  additional information about $\varphi_1$ in the
  singular case   $abcd=1$.

For $q=0$, Ref. \cite{Askey-Wilson-85} relates Askey-Wilson polynomials  $p_n$ to the Chebyshev polynomials $U_n$ of second kind.
Denote by $s_j(a,b,c,d)$ the $j$-th symmetric function, i.e. $s_1=a+b+c+d$, $s_2=ab+ac+ad+bc+bd+cd$, $s_3=abc+abd+acd+bcd$, $s_4=abcd$. Then with $U_{-1}=0$ we have
\begin{eqnarray*}
p_0&=&U_0 \\
 p_1&=&(1-s_4) U_1+(s_3-s_1)U_0
 \\
  p_2&=&U_2-s_1U_1+(s_2-s_4)U_0 \\
p_n&=&U_n-s_1U_{n-1}+s_2U_{n-2}-s_3U_{n-3}+s_4U_{n-4} \mbox{ for } n\geq 3.
\end{eqnarray*}
Recall that $G_n(1)=\varphi_0[H_n(\ee+\dd)]=\varphi_0[U_n(\xx)]$.
So in the non-singular case the conclusion of Theorem \ref{T3} follows from the following relations between $G_n(1)$.
\begin{eqnarray}
  (1 - s_4) G_1(1) + (s_3 - s_1) G_0(1)&=&0 \label{Eqnt-1}\\
  G_2(1)-
 s_1 G_1(1)+ (s_2 - s_4)  G_0(1) &=&0 \label{Eqnt-2}\\
 G_n(1) - s_1 G_{n-1}(1) +
 s_2 G_{n-2}(1) - s_3 G_{n-3}(1) +
 s_4 G_{n-4}(1)&=&0, \quad \; n\geq 3 \label{Eqnt-n}.
\end{eqnarray}
These relations can be established by analyzing   explicit solutions of recursion \eqref{Eq:MArcin2}.
  We first determine the initial (irregular) solutions
 $$G_1(t)=\frac{(c+d) (t-a b)+(a+b) (1-c d t)}{1-a b c d}$$  and
 \[G_2(t)=t (c+d)G_1(t)  +\frac{(a+b) (1-c d t) (a+b- a b (c+d))}{1-a b c d} -a b (
   1-cd t)-c d t^2\]
   which we use with $t=1$ to verify \eqref{Eqnt-1} and \eqref{Eqnt-2}.
Next, we use \eqref{Eq:MArcin2} with $t=0$ and $n\geq 1$ to determine
 $\alpha_n=G_n(0)$ from the  recursion of order $2$,
 \begin{equation}\label{G0}
 \alpha_{n+1}(0)=(a+b)\alpha_n -ab \alpha_{n-1} .
 \end{equation}
 Since in our setting arising from ASEP parameters  $b\leq 0<a$ are not equal, the general solution  is
 $$\alpha_n=C_1 a^n+C_2 b^n.$$
 The constants $C_1, C_2 $
 are determined from the initial values of $G_0(0)=1$ and $G_1(0)=\frac{a+b-a b (c+d)}{1-a b c d}$. We get
 $$\alpha_n=
 \frac{ (
 1-b  c) (1-b d)}{(a-b) (1-a b c d)}a^{n+1}+ \frac{(1-a c) (1-a d)}{(b-a) (1-a b c d)} b^{n+1}.$$
Next we solve the recursion for $z_n=G_n(1)$. This is now a non-homogeneous recursion
$$
z_{n+1} =(1-cd)((a+b)\alpha_n-ab \alpha_{n-1})+(c+d)z_n -cd z_{n-1},
$$
which we simplify using \eqref{G0} into
$$
z_{n+1} =(1-cd)\alpha_{n+1}+(c+d)z_n-cd z_{n-1}.
$$
Since $d\leq 0<c$, the general solution of this recursion is
$$G_n(1)=z_n= B_1 a^{n+3}+B_2 b^{n+3}+K_1 c^{n+3}+K_2 d^{n+3}, \; n\geq 0$$
where
$$B_1=\frac{(1-b c) (1-b d) (1-c d)}{(a-b) (a-c) (a-d) (1-a b c d)}, \; \quad B_2=\frac{(1-a c) (1-a d) (1-c d)}{(b-a) (b-c) (b-d) (1-a b c d)}$$
come from the undetermined coefficients method and
$$ K_1=\frac{(1-a b) (1-a d) (1-b d) }{(c-a)
   (c-b) (c-d) (1-a b c d)} , \; \quad K_2=\frac{ (1-a b) (1-a c) (1-b c)
 }{(d-a) (d-b) (d-c) (1-a b c d)} $$
 come from matching the initial values. It
turns out that the explicit values of the constants are only needed for verification of the initial equations, as equation
\eqref{Eqnt-n} holds for any linear combination of $a^n,b^n,c^n,d^n$.

Proceeding in  similar way  we can also derive a version of Theorem \ref{T3}  that  relates functional $\varphi_1$ to Askey-Wilson polynomials.
We have
$$F_0(t) =0, \quad
F_1(t)=\frac{1-cdt}{1-cd}.$$
The recursion for $\alpha_n=F_n(0)$ is \eqref{G0}, so using the above initial values we get the solution
$$
F_n(0)=\frac{a^{n}-b^{n}}{(a-b)(1-cd)}, \quad  n\geq0.
$$
The recursion for $F_n(1)$ is
$$F_{n+1}(1)=(c+d)F_n(1)-c d F_{n-1}(1)+\frac{a^n-b^n}{a-b}, \quad n\geq 1.$$
Here the constants are simpler and a  calculation gives
\begin{multline}
  \label{Fq=0}
  F_{n}(1)= \frac{a^{n+2}}{(a-b)(a-c)(a-d)}+ \frac{b^{n+2}}{(b-a)(b-c)(b-d)}+\frac{c^{n+2}}{(c-a)(c-b)(c-d)}
  \\+\frac{d^{n+2}}{(d-a)(d-b)(d-c)},\; n\geq0 .
\end{multline}
Noting that in the singular case $p_1$ is a constant, we have $\varphi_1[p_n(\xx)]=0$ for all $n=0,1,\dots$.
\arxiv{ To avoid the irregularity with $p_1$ in the singular case, we can also consider the following family of polynomials:
\begin{eqnarray*}
q_0(x)&=&U_0(x)\\
  q_1(x)&=&U_1(x)+(s_3-s_1)U_0 (x)\\
  q_2(x)&=& U_2(x)-s_1 U_1(x)+(s_2-s_4) U_0(x) \\
  q_n(x)&=& U_n(x)-s_1U_{n-1}(x)+s_2U_{n-2}(x)-s_3U_{n-3}(x)+s_4U_{n-4}(x), \quad n\geq 3.
\end{eqnarray*}
Since $2xU_n=U_{n+1}+U_{n-2}$, polynomials $q_n$ satisfy the following finite perturbation of the constant three step recursion:
\begin{eqnarray*}
   2x q_0&=&q_1+(s_1-s_3)q_0 \\
   2x q_1&=&q_2+s_3 q_1+(s_4-s_2+s_3(s_1-s_3))q_0\\
   2x q_2&=&q_3+q_1 \\
   2x q_n&=&q_{n+1}+q_{n-1}, \quad n\geq 2.
\end{eqnarray*}
As previously, \eqref{Fq=0} implies that $\varphi_1[q_1(\xx)]=1$ and $\varphi_1[q_n(\xx)]=0$  for $n\geq 2$.
Since  $x^kq_n$ is a linear combination of $g_{n-k},g_{n-k+1},\dots,g_{n+k}$ this implies that
$$
\varphi_1[q_k(\xx)q_n(\xx)]=0 \mbox{ for } |n-k|\geq 2.
$$

}

Motivated by the generating function  $\sum_{n=0}^\infty H_n(x)z^n=1/(1+z^2-xz)$ lets denote by
$\varphi[(1+z^2 -(\ee+\dd) z)^{-1}]$ the   power series $\sum_{n=0}^\infty \varphi[H_n(\ee+\dd)]z^n$.  We can now summarize the above formulas more concisely.
\begin{proposition}\label{P-Szpojankowski} If $abcd\ne 1$ then for $|z|$ small enough
  \[\varphi_0[(1+z^2 -(\ee+\dd) z)^{-1}]=\frac{1+z^2 abcd}{(1-az)(1-bz)(1-cz)(1-dz)} +
  \frac{ z abcd(a+b+c+d-(1/a+1/b+1/c+1/d))  }{(1-a b c d)(1-az)(1-bz)(1-cz)(1-dz)}.\]
  If $abcd= 1$ then for $|z|$ small enough
$$
      \varphi_1[(1+z^2 -(\ee+\dd) z)^{-1}]=\frac{z}{(1-a z) (1-b z) (1-c z) (1-d z)}.
$$
\end{proposition}
The first expression matches the formula from \cite[Theorem 4.1]{Szpojankowki2010}   who computed the integral of $1/(1+z^2-xz)$ with respect to  the Askey-Wilson measure with $q=0$  under the assumptions which in our setting  boil down to
$ac \leq 1$ and $abcd<1$.
\arxiv{
Indeed, \[\varphi_0[(1+z^2 -(\ee+\dd) z)^{-1}]=\sum_{n=0}^\infty z^n G_n(1)=
\frac{1+z^2 abcd}{(1-az)(1-bz)(1-cz)(1-dz)} +
  \frac{ z abcd(a+b+c+d-(1/a+1/b+1/c+1/d))  }{(1-a b c d)(1-az)(1-bz)(1-cz)(1-dz)}\]
and from \eqref{Fq=0} we get 
\[
\varphi_1[(1+z^2 -(\ee+\dd) z)^{-1}]=\sum_{n=0}^\infty z^n F_n(1)=\frac{z}{(1-a z) (1-b z) (1-c z) (1-d z)}
\]
}

 \section{A matrix model}\label{Sec:MatrixModel}
According to   \citet{Mallick-Sandow-1997}
stationary probabilities for ASEP with large $L $ can be computed from a finite matrix model when the parameters
 satisfy condition $q^{m}ac=1$ for some $m\geq 0$. Here we present a version of this model, together with a caution about a subtle issue that may affect some infinite matrix models.

Recalling that in \eqref{AW-parameters} we chose $a>0$, for $q> 0$ we  consider two infinite matrices
\begin{equation}
\mE=\theta^2 \begin{bmatrix}
   1+\frac{1}{a} & 0  & 0&   & \dots &0& \dots\\
   1 & 1+ \frac{1}{ a q}  & 0& & \dots &  \\
 0& 1 &\ddots & &  &   \\
    \vdots   &   &\ddots &&\ddots&     \vdots \\
   \\  0  & 0&\dots&0&  1& 1+\frac{1}{a q^{n-1}} \\ \\
   \vdots &&&& & \ddots&\ddots
\end{bmatrix} \; \mD=\theta^2 \begin{bmatrix}
   1+a & 0 & 0 & \dots &  \\
   0 & 1+a q &0& \dots &  \\
   \vdots&\vdots &\ddots & &
   \\  0& 0 &\dots &  1+a q^{n-1} \\ \\
   \vdots &\vdots&&& \ddots
\end{bmatrix}
\end{equation}
It is straightforward to verify that identity \eqref{q-comm-Derrida} is satisfied. Conditions \eqref{W} and \eqref{V} become recursions for the components of the vectors
$$\langle W|=[w_1,w_2,\dots] \mbox{ and } |V\rangle =[v_1,v_2,\dots]^T.$$
In  parametrization \eqref{AW-parameters}, conditions \eqref{W} and \eqref{V} become \eqref{WeV} and \eqref{WdV}, and the
resulting recursions are
$$
\frac{1}{a q^{k-1}} w_k+w_{k+1}=(c+d)w_k-a cd q^{k-1}w_k,
$$
$$
ab \left(v_{k-1}+\frac{1}{a q^{k-1}}v_k\right)=(a+b)v_k-a q^{k-1} v_k.
$$

 \arxiv{Conditions \eqref{W} and \eqref{V} are $(1-q)\langle W|(\mE+c d
\mD)=(1+c)(1+d)\langle W|$ and $(1-q)(ab\mE+\mD)|V\rangle=(1+a)(1+b)|V\rangle$.

To derive \eqref{WeV} and \eqref{WdV}, we insert \eqref{DE2de} into the above equations, and simplify the expressions.

 To derive  the recursions as written above, we  compute
 $$
\dd=\theta\begin{bmatrix}
  a & 0  & 0&   & \dots & &  \\
   0 & a q  & 0& & \dots &  \\
 0& 0 &a q^2 & &  &   \\
    \vdots   &   &\ddots &&\ddots&       \\
   \\  0  & 0&\dots&0&  & a q^{k-1} \\ \\
   \vdots &&&&  \ddots&\ddots
\end{bmatrix}, \quad
\ee=\theta\begin{bmatrix}
   \frac{1}{a} & 0  & 0&   & \dots & &  \\
  1 & \frac{1}{a q}  & 0& & \dots &  \\
 0& 1 &\frac{1}{a q^2} & &  &   \\
    \vdots   &   &\ddots &&\ddots&      \\
   \\  0  & 0&\dots&0&  1& \frac{1}{a q^{k-1}} \\ \\
   \vdots &&&& \ddots&\ddots
\end{bmatrix}.
$$

 } With $w_1=v_1=1$, the solutions are explicit
\begin{equation}\label{w}
w_{n}= \prod_{k=1}^{n-1}\left(c+d -acd q^{k-1}-\frac{1}{a q^{k-1}}\right)= \frac{(a c,ad;q)_{n-1}}{(-a)^{n-1}q^{(n-1)(n-2)/2}},
\end{equation}
\begin{equation}
  \label{v}
  v_n= \frac{  a^{n-1}b^{n-1}}{\prod_{k=1}^{n-1}\left(a (1-q^k)+b(1-1/q^k)\right)}
= \frac{(-a)^{n-1}q^{n(n-1)/2}}{(q, qa/b;q)_{n-1}}.
\end{equation}
\arxiv{We remark that since $a>0$ and $b\leq 0$ the second expression for $v_n$ is well defined only if $b<0$, i.e,. when $\delta>0$, see \eqref{AW-parameters}.
 When $b=0$,  from the first expression we get  $V=[1,0,0,\dots]^T$, and the formulas we discuss below are not valid.}

 We therefore get  explicit formula
\begin{equation*}
  \label{WV}
\langle W |\mI|V\rangle =  \sum_{k=1}^{\infty} v_kw_k %
= \sum_{k=1}^{\infty} q^{k-1}\frac{(ac,ad;q)_{k-1}} {(q,a q/b;q)_{k-1}} ={_2\phi_1}\left(\begin{matrix}
 ac,\ a d \\\ q  a /b
\end{matrix}\middle|q;q\right),
\end{equation*}
  valid for  $0<q<1$.
Somewhat more generally, since $\dd$  in \eqref{DE2de} becomes a diagonal matrix with the sequence $\{\theta a q^{k-1}\}$
on the diagonal, we get
\begin{equation}
  \label{Wd^LV}
  \langle W |\dd^L|V\rangle =a^L\theta^L {_2\phi_1}\left(\begin{matrix}
 ac,\ a d \\\ q  a /b
\end{matrix}\middle|q;q^{L+1}\right).
\end{equation}
(We will use this formula for $L=0,1$ in Section \ref{Sec:Warn}.)

We now consider the case  when parameters $a,c$   are such that
 $ac q^m=1$ for some integer $m\geq0$.
In this case the infinite series  terminate as  formula \eqref{w} gives $w_n=0$ for all $n\geq m+2$.
Since each monomial $\mX$ is a lower-triangular matrix, in this case     components $v_k$ with $k\geq m+2$  do not enter the calculation of $\langle W|\mX|V\rangle$, so
  we can truncate $\ee, \dd, \mI$ to their $m+1$ by $m+1$ upper left corners,  recovering a version of
the finite matrix model from \citet{Mallick-Sandow-1997}.

Using \eqref{twophione}
one can show that
$$
{_2\phi_1}\left(\begin{matrix}
  q^{-m},\ a d \\\  q a /b
\end{matrix}\middle|q;q\right)=
\frac{(bdq^{-m};q)_m}{(bc;q)_m}.
$$
\arxiv{Applying transformation  \eqref{twophione}   we rewrite ${_2\phi_1}\left(\begin{matrix}
  q^{-m},\ a d \\\  q a /b
\end{matrix}\middle|q;q\right)$ as
\begin{multline*}
  (ad)^m \frac{(q/(bd);q)_m}{(a q/b;q)_m}
=q^{-m^2} \frac{(q-bd)(q^2-bd)\dots(q^{m}-bd)}{(1-bc)(q^{-1}-bc)\dots (q^{1-m}-bc)}
\\=q^{-m^2} \frac{q^{m(m+1)/2}(1-bd/q)(1-bd/q^2)\dots(1-bd/q^m)}{q^{-m(m-1)/2}(1-bc)(1-q bc)\dots (1-q^{m-1}bc)}
=\frac{(bdq^{-m};q)_m}{(bc;q)_m}.
\end{multline*}

} Thus, in agreement with findings in \citet{Mallick-Sandow-1997},   $$\langle W |\mI|V\rangle=\frac{(bdq^{-m};q)_m}{(bc;q)_m}$$ vanishes
if and only if $bd\in\{q,q^2,\dots,q^{m}\}$, i.e., in the singular case when $q^Nabcd=1$ for some $N=0,\dots,m-1$.
One would
expect that in this case the matrix model should be related to functional $\varphi_1$ by a simple renormalization but we have
not verified the details.

In the non-singular case (but still with $q^m ac=1$)  the relation   is  straightforward.  Due to shared recursion and
initialization at $\mI$, it is clear that  functional $\varphi_0$ is indeed related to the matrix model by
\begin{equation*}
  \label{phi0Mallick}
 \langle W | \mX|V\rangle=  \frac{(bdq^{-m};q)_m}{(bc;q)_m} \varphi_0[\mX].
\end{equation*}

\begin{remark}
From the reviewer report we learned that Refs. \cite{Krebs_2003} and \cite{Jafarpour_2007} relate the  finite-dimensional representations from \citet{Mallick-Sandow-1997} %
  to  convex combinations of Bernoulli shock measures with $m$ shocks. It would be interesting to see how this is reflected in the
structure of   functional $\varphi$.
\end{remark}

A natural question then arises how the functionals $\varphi_0$, or $ \varphi_1$, are related to this matrix model for more
general parameters $a,b,c,d$. The surprising answer  is that there is no such relation, as we explain next.

\subsection{A caution about matrix models}\label{Sec:Warn}
It is known, \citep{bossaller2019associativity,keremedis1988associativity}, but perhaps this is not appreciated enough, that
 multiplication of infinite matrices may fail to be associative for other reasons than  divergence. And precisely
this difficulty afflicts the above matrix model when $acq^n\ne1$ for all $n$.
To see the source of the difficulty, we rewrite
 \eqref{W} and \eqref{V}
  as
\begin{equation}
  \label{WeV} \langle W| \ee=\theta (c+d)\langle W|- c d \langle W|\dd,
\end{equation}
\begin{equation}
  \label{WdV}  ab \ee|V\rangle=\theta(a+b)|V\rangle - \dd|V\rangle.
\end{equation}
To indicate  clearly the order of matrix multiplications, lets denote vector $\langle W| \ee$ by $\langle \tilde W|$  and
vector $\ee|V\rangle$ by $|\tilde V\rangle$. Using    \eqref{WeV} and \eqref{WdV},   we could  compute the product
$\langle W|\ee|V\rangle$  of three matrices either as $\langle \tilde W|V\rangle$, or as $\langle W |\tilde V\rangle$. From the
first calculation we get
$$\langle \tilde W|V\rangle   %
=\theta (c+d) {_2\phi_1}\left(\begin{matrix}
  q^{-m},\ a d \\\  q a /b
\end{matrix}\middle|q;q\right) - a c d \theta {_2\phi_1}\left(\begin{matrix}
  q^{-m},\ a d \\\  q a /b
\end{matrix}\middle|q;q^2\right)$$
where %
we used \eqref{Wd^LV} with $L=0$ and $L=1$ on the right hand side. The second calculation gives a different answer
$$ab \langle W |\tilde V\rangle %
=\theta (a+b) {_2\phi_1}\left(\begin{matrix}
  ac,\ a d \\\  q a /b
\end{matrix}\middle|q;q\right) - a   \theta {_2\phi_1}\left(\begin{matrix}
  ac,\ a d \\\  q a /b
\end{matrix}\middle|q;q^2\right).$$
In fact, we have
$$\langle\tilde W|V\rangle=\theta \sum_{k=1}^\infty \left(\frac{1}{a q^{k-1}} w_k+w_{k+1}\right)v_k$$
$$\langle W|\tilde V\rangle=\theta \sum_{k=1}^\infty w_k\left(v_{k-1}+\frac{1}{a q^{k-1}} v_k\right) \mbox{ with $v_{-1}=0$}.$$
So from \eqref{w} and \eqref{v} we get
$$\langle \tilde W | V\rangle-\langle W|\tilde V\rangle= \lim_{n\to\infty}\sum_{k=1}^n(w_{k+1}v_k-w_kv_{k-1})=
\lim_{n\to\infty} w_{n+1}v_n =-\frac{\theta}{a}\frac{ (ac,ad;q)_\infty}{ (q,qa/b;q)_\infty}.
$$
This shows that in general multiplication of matrices $\langle W|$, $\ee$ and $|V\rangle$ is not associative. Since $d\leq 0$,
the two answers match only when $q^mac=1$ for some $m$, i.e., in the terminating case. This is precisely the case considered by
\cite{Mallick-Sandow-1997}, and of course multiplication of finite dimensional matrices is associative.

 \arxiv{This  established  the following hypergeometric function identity
 \begin{equation*}
  \label{FasleId}
  a(1-abcd){_2\phi_1}\left(\begin{matrix}
 ac,\ a d \\\  q a /b
\end{matrix}\middle|q;q^2\right) =(a+b-ab(c+d)){_2\phi_1}\left(\begin{matrix}
  ac,\ a d \\\  q a /b
\end{matrix}\middle|q;q\right)+b \frac{ (ac,ad;q)_\infty}{ (q,qa/b;q)_\infty}.
\end{equation*}

}

%

\begin{thebibliography}{}

\bibitem[Aneva, 2007]{aneva2007matrix}
Aneva, B. (2007).
\newblock Matrix-product ansatz as a tridiagonal algebra.
\newblock {\em Journal of Physics A: Mathematical and Theoretical},
  40(39):11677.

\bibitem[Aneva, 2009]{Aneva-2009-Integrability}
Aneva, B. (2009).
\newblock Integrability condition on the boundary parameters of the asymmetric
  exclusion process and matrix product ansatz.
\newblock In {\em Trends in differential geometry, complex analysis and
  mathematical physics}, pages 10--19. World Sci. Publ., Hackensack, NJ.

\bibitem[Askey and Wilson, 1985]{Askey-Wilson-85}
Askey, R. and Wilson, J. (1985).
\newblock Some basic hypergeometric orthogonal polynomials that generalize
  {J}acobi polynomials.
\newblock {\em Mem. Amer. Math. Soc.}, 54(319):iv+55.

\bibitem[Bjorner and Brenti, 2005]{bjorner2006combinatorics}
Bjorner, A. and Brenti, F. (2005).
\newblock {\em Combinatorics of {C}oxeter groups}, volume 231 of {\em Graduate
  Texts in Mathematics}.
\newblock Springer Science \& Business Media.

\bibitem[Bossaller and L{\'o}pez-Permouth, 2019]{bossaller2019associativity}
Bossaller, D.~P. and L{\'o}pez-Permouth, S.~R. (2019).
\newblock On the associativity of infinite matrix multiplication.
\newblock {\em The American Mathematical Monthly}, 126(1):41--52.

\bibitem[Bo{\.z}ejko et~al., 1997]{bozejko97qGaussian}
Bo{\.z}ejko, M., K{\"u}mmerer, B., and Speicher, R. (1997).
\newblock {$q$}-{G}aussian processes: non-commutative and classical aspects.
\newblock {\em Comm. Math. Phys.}, 185(1):129--154.

\bibitem[Bryc and Weso{\l}owski, 2017]{Bryc-Wesolowski-2015-asep}
Bryc, W. and Weso{\l}owski, J. (2017).
\newblock {A}symmetric {S}imple {E}xclusion {P}rocess with open boundaries and
  {Q}uadratic {H}arnesses.
\newblock {\em Journal of Statistical Physics}, 167:383--415.

\bibitem[Chihara, 2011]{chihara2011introduction}
Chihara, T.~S. (2011).
\newblock {\em An introduction to orthogonal polynomials}.
\newblock Courier Corporation.

\bibitem[Corteel and Williams, 2011]{corteel2011tableaux}
Corteel, S. and Williams, L.~K. (2011).
\newblock Tableaux combinatorics for the asymmetric exclusion process and
  {A}skey-{W}ilson polynomials.
\newblock {\em Duke Mathematical Journal}, 159(3):385--415.

\bibitem[Derrida et~al., 1992]{derrida1992exact}
Derrida, B., Domany, E., and Mukamel, D. (1992).
\newblock An exact solution of a one-dimensional asymmetric exclusion model
  with open boundaries.
\newblock {\em Journal of Statistical Physics}, 69(3-4):667--687.

\bibitem[Derrida et~al., 1993]{derrida1993exact}
Derrida, B., Evans, M., Hakim, V., and Pasquier, V. (1993).
\newblock Exact solution of a {1D} asymmetric exclusion model using a matrix
  formulation.
\newblock {\em Journal of Physics A: Mathematical and General},
  26(7):1493--1517.

\bibitem[Derrida and Mallick, 1997]{derrida-mallick-97}
Derrida, B. and Mallick, K. (1997).
\newblock Exact diffusion constant for the one-dimensional partially asymmetric
  exclusion model.
\newblock {\em J. Phys. A}, 30(4):1031--1046.

\bibitem[Enaud, 2005]{Enaud2005}
Enaud, C. (2005).
\newblock {\em Processus d'exclusion asym\'etrique: Effet du d\'esordre,
  Grandes d\'eviations et fluctuations}.
\newblock PhD thesis, Universit\'ee Pierre et Marie Curie - Paris VI.

\bibitem[Enaud and Derrida, 2004]{enaud2004large}
Enaud, C. and Derrida, B. (2004).
\newblock Large deviation functional of the weakly asymmetric exclusion
  process.
\newblock {\em Journal of Statistical Physics}, 114(3-4):537--562.

\bibitem[Essler and Rittenberg, 1996]{essler1996representations}
Essler, F.~H. and Rittenberg, V. (1996).
\newblock Representations of the quadratic algebra and partially asymmetric
  diffusion with open boundaries.
\newblock {\em Journal of Physics A: Mathematical and General},
  29(13):3375--3407.

\bibitem[Foupouagnigni et~al., 2013]{foupouagnigni2013connection}
Foupouagnigni, M., Koepf, W., and Tcheutia, D. (2013).
\newblock Connection and linearization coefficients of the {A}skey--{W}ilson
  polynomials.
\newblock {\em Journal of Symbolic Computation}, 53:96--118.

\bibitem[Frisch and Bourret, 1970]{frisch1970parastochastics}
Frisch, U. and Bourret, R. (1970).
\newblock Parastochastics.
\newblock {\em Journal of Mathematical Physics}, 11(2):364--390.

\bibitem[Gasper and Rahman, 2004]{gasper2004basic}
Gasper, G. and Rahman, M. (2004).
\newblock {\em Basic hypergeometric series}.
\newblock Cambridge University Press.

\bibitem[Gorissen et~al., 2012]{gorissen2012exact}
Gorissen, M., Lazarescu, A., Mallick, K., and Vanderzande, C. (2012).
\newblock Exact current statistics of the asymmetric simple exclusion process
  with open boundaries.
\newblock {\em Physical review letters}, 109(17):170601.

\bibitem[Jafarpour and Masharian, 2007]{Jafarpour_2007}
Jafarpour, F.~H. and Masharian, S.~R. (2007).
\newblock Matrix product steady states as superposition of product shock
  measures in 1d driven systems.
\newblock {\em Journal of Statistical Mechanics: Theory and Experiment},
  2007(10):P10013--P10013.

\bibitem[Keremedis and Abian, 1988]{keremedis1988associativity}
Keremedis, K. and Abian, A. (1988).
\newblock On the associativity and commutativity of multiplication of infinite
  matrices.
\newblock {\em International Journal of Mathematical Education in Science and
  Technology}, 19(1):175--197.

\bibitem[Koekoek and Swarttouw, 1998]{koekoek1998askey}
Koekoek, R. and Swarttouw, R. (1998).
\newblock The {A}skey-scheme of hypergeometric orthogonal polynomials and its
  qanalogue, online at http://aw. twi. tudelft. nl/\~{} koekoek/askey. html,
  report 98-17.
\newblock {\em Technical University Delft}, 2:20--21.

\bibitem[Krebs et~al., 2003]{Krebs_2003}
Krebs, K., Jafarpour, F.~H., and Sch\"utz, G.~M. (2003).
\newblock Microscopic structure of travelling wave solutions in a class of
  stochastic interacting particle systems.
\newblock {\em New Journal of Physics}, 5:145--145.

\bibitem[Lazarescu, 2013]{lazarescu2013matrix}
Lazarescu, A. (2013).
\newblock Matrix ansatz for the fluctuations of the current in the {ASEP} with
  open boundaries.
\newblock {\em Journal of Physics A: Mathematical and Theoretical},
  46(14):145003.

\bibitem[Lemay et~al., 2018]{lemay2018q}
Lemay, J.-M., Vinet, L., and Zhedanov, A. (2018).
\newblock A $q$-generalization of the para-{R}acah polynomials.
\newblock {\em Journal of Mathematical Analysis and Applications},
  462(1):323--336.

\bibitem[Liggett, 1975]{Liggett-1975}
Liggett, T.~M. (1975).
\newblock Ergodic theorems for the asymmetric simple exclusion process.
\newblock {\em Trans. Amer. Math. Soc.}, 213:237--261.

\bibitem[Mallick and Sandow, 1997]{Mallick-Sandow-1997}
Mallick, K. and Sandow, S. (1997).
\newblock Finite-dimensional representations of the quadratic algebra:
  applications to the exclusion process.
\newblock {\em J. Phys. A}, 30(13):4513--4526.

\bibitem[Paule and Riese, 1997]{paule1997mathematica}
Paule, P. and Riese, A. (1997).
\newblock A {M}athematica $q$-analogue of {Z}eilberger's algorithm based on an
  algebraically motivated approach to $q$-hypergeometric telescoping.
\newblock {\em Special functions, $q$-series and related topics}, 14:179--210.

\bibitem[Sandow, 1994]{sandow1994partially}
Sandow, S. (1994).
\newblock Partially asymmetric exclusion process with open boundaries.
\newblock {\em Physical Review E}, 50(4):2660--2667.

\bibitem[Szpojankowski, 2010]{Szpojankowki2010}
Szpojankowski, K. (2010).
\newblock Free quadratic harnesses.
\newblock Master's thesis, Warsaw University of Technology.
\newblock (In Polish).

\bibitem[Tcheutia, 2014]{tcheutia2014connection}
Tcheutia, D.~D. (2014).
\newblock {\em On connection, linearization and duplication coefficients of
  classical orthogonal polynomials}.
\newblock PhD thesis, University of Kassel, Germany.

\bibitem[Tsujimoto et~al., 2017]{tsujimoto2017tridiagonal}
Tsujimoto, S., Vinet, L., and Zhedanov, A. (2017).
\newblock Tridiagonal representations of the $q$-oscillator algebra and
  {A}skey--{W}ilson polynomials.
\newblock {\em Journal of Physics A: Mathematical and Theoretical},
  50(23):235202.

\bibitem[Uchiyama et~al., 2004]{uchiyama2004asymmetric}
Uchiyama, M., Sasamoto, T., and Wadati, M. (2004).
\newblock Asymmetric simple exclusion process with open boundaries and
  {A}skey--{W}ilson polynomials.
\newblock {\em Journal of Physics A: Mathematical and General},
  37(18):4985--5002.

\end{thebibliography}
%
%
%

\end{document}